%% file: implementation-v2.tex
\documentclass[a4paper,american,reqno]{amsart}

\usepackage{amsmath,amsthm,amssymb}
\providecommand{\burl}[1]{\url{#1}}
\input{setup-preprint}
 \bibliography{implementation.bib}
\begin{document}
	
	\title[Minimal Regret Walrasian Equilibria for Combinatorial Markets]%
	{Minimal Regret Walrasian Equilibria\\ for Combinatorial Markets}
	
	\author[A. Duguet, T. Harks, M. Schmidt, J. Schwarz]%
	{Aloïs Duguet, Tobias Harks, Martin Schmidt, Julian Schwarz}
	
	\address[A. Duguet, M. Schmidt]{%
		Trier University,
		Department of Mathematics,
		Universitätsring 15,
		54296 Trier,
		Germany}
	\email{duguet@uni-trier.de}
	\email{martin.schmidt@uni-trier.de}
	
	\address[T. Harks, J. Schwarz]{%
		University of Passau,
		Faculty of Computer Science and Mathematics,
		94032 Passau,
		Germany}
	\email{tobias.harks@uni-passau.de}
	\email{julian.schwarz@uni-passau.de}
	
	\date{\today}
	
	\begin{abstract}
		\input{abstract-v2}
	\end{abstract}

\keywords{\input{keywords}}

\maketitle

\input{introduction-v2}
\input{model-v2}
\input{general-valuations-v2}
\input{poly-time-approx-schemes}

\input{lower-bounds}

\input{applications}
\input{acknowledgements}
\input{use-of-ai}


\printbibliography
\appendix
\input{appendix}

\end{document}

%% file: setup-preprint.tex
\usepackage{babel}
\usepackage{lmodern}
\usepackage{xcolor} 
\usepackage[utf8]{inputenc}
\usepackage[T1]{fontenc}
\usepackage{csquotes}
\usepackage{todonotes}
\usepackage{tikz}
\usepackage{enumitem}
\usepackage{algorithm,algpseudocode}
\usepackage[style = authoryear-comp,
            maxbibnames = 100,
            maxcitenames = 2,
            giveninits = true,
            uniquename = init,
            isbn = false,
            backend = bibtex]{biblatex}
\usepackage[colorlinks,
            citecolor=blue,
            urlcolor=blue,
            linkcolor=blue]{hyperref}
\usepackage{nicefrac}
\usepackage{array} 
\defbibheading{bibliography}[\refname]{\section*{#1}}

\makeatletter
\patchcmd{\@settitle}{\uppercasenonmath\@title}{\scshape\large}{}{}
\patchcmd{\@setauthors}{\MakeUppercase}{\scshape\normalsize}{}{}
\makeatother

\newtheorem{theorem}{Theorem}[section]
\newtheorem{example}[theorem]{Example}
\newtheorem{remark}[theorem]{Remark}
\newtheorem{lemma}[theorem]{Lemma}
\newtheorem{proposition}[theorem]{Proposition}
\newtheorem{definition}[theorem]{Definition}
\newtheorem{corollary}[theorem]{Corollary}

\newtheorem{innerrsttheorem}{Theorem}
\newenvironment{rsttheorem}[1]
{\renewcommand\theinnerrsttheorem{#1}\innerrsttheorem}
{\endinnerrsttheorem}

\newtheorem{innerrstcorollary}{Corollary}
\newenvironment{rstcorollary}[1]
{\renewcommand\theinnerrstcorollary{#1}\innerrstcorollary}
{\endinnerrstcorollary}

\newtheorem{innerrstexample}{Example}
\newenvironment{rstexample}[1]
{\renewcommand\theinnerrstexample{#1}\innerrstexample}
{\endinnerrstexample}
\vfuzz \hfuzz
\setlist[enumerate]{label=(\alph*)}

\makeatletter
\@namedef{subjclassname@2020}{%
  \textup{2020} Mathematics Subject Classification}
\makeatother

\input{macros}


%% file: macros.tex
\DeclareMathOperator*{\argmax}{arg\,max}
\newcommand{\inte}{\mathrm{int}}
\newcommand{\ALG}{\mathrm{ALG}}

\newcommand{\con}{\mathrm{con}}

\newcommand{\LP}{\mathrm{LP}}

\newcommand{\val}{\mathrm{val}}

\newcommand{\Reg}{\mathrm{Reg}}

\newcommand{\R}{\mathbb{R}}
\newcommand{\Z}{\mathbb{Z}}
\newcommand{\N}{\mathbb{N}}

\NewDocumentCommand{\freeset}{O{}}{S^{#1}}

\NewDocumentCommand{\FeasStrats}{O{}}{W_{#1}}
\NewDocumentCommand{\FeasStratsInt}{O{}}{W_{#1}^\inte}
\NewDocumentCommand{\FeasStratsCon}{O{}}{W_{#1}^\con}
\NewDocumentCommand{\FeasStratsRel}{O{}}{\hat{W}_{#1}}

\newcommand{\Simp}{A}
\newcommand{\Rleq}{R^<}
\newcommand{\Req}{R^=}

\newcommand{\Walras}{Walrasian }

\newcommand{\wrt}{w.r.t.\ }

\NewDocumentCommand{\NeInt}{O{\ }}{TODO{#1}}

\newcommand{\defset}[3][\defsep]{\set{#2#1#3}}
\newcommand{\Defset}[3][\defsep]{\Set{#2#1#3}}

\newcommand{\Set}[1]{\left\{#1\right\}}
\newcommand{\st}{\text{s.t.}}
\newcommand{\define}{\mathrel{{\mathop:}{=}}}


%% file: abstract-v2.tex
We consider combinatorial multi-item markets
and propose the notion of a $\Delta$-\emph{regret \Walras equilibrium},
consisting of an allocation and item prices such that
prices clear the market, the allocation is capacity-feasible, and the
players' strategies lead to a total regret of at most
$\Delta$. Regret is the sum of the players' utility gaps
  relative to optimal bundles at the given prices.
We derive a complete characterization for the existence of
$\Delta$-regret equilibria by introducing the concept of a
\emph{parameterized} social welfare problem,
where only the capacity constraints binding at a prescribed
  allocation are retained.
Our characterization then relates the achievable regret value
to the associated
duality/integrality gap of the parameterized social welfare
problem at that allocation.
For the special case of monotone valuations and
  upwards-closed strategy spaces this translates to regret bounds
recovering the duality/integrality gap of the original social welfare
problem.
We also characterize weighted buyer and seller regret by
configuration-LP gaps and establish a finite exact penalty enforcing
market clearing. We further derive a slack-dependent price-adjustment
bound whose uniform version is sharp for prescribed allocations.
Finally, these structural results translate approximation
algorithms for social welfare with guarantees relative to the configuration LP
into algorithms computing low-regret \Walras equilibria.
We also derive lower bounds from integrality and duality gaps
and computational complexity.


%% file: keywords.tex
Combinatorial markets,
\Walras equilibria,
Regret equilibria,
Duality gaps,
Integrality gaps,
Market clearing%
%
%

%% file: introduction-v2.tex
\section{Introduction}

\Walras market equilibria \parencite{walras54} constitute
a central topic in economics \parencite{Shapley1971}, computer science
\parencite{Blumrosen_Nisan_2007}, and mathematics \parencite{Murota:2003}.
Combinatorial markets constitute an important subclass, where a set $R=\{1,\dots,m\}$ of indivisible items is available at integer
multiplicity $u=(u_j)_{j\in R}$.
There is a set of players $N=\{1,\ldots,n\}$ with
valuations over bundles of these items given by
$\pi_i:X_i\rightarrow \R$, where $X_i\subset \Z_+^m$
is the set of available bundles (strategies).
The \emph{\Walras equilibrium} for quasi-linear utilities asks for
(anonymous) per-unit item prices \mbox{$\lambda_j \geq 0$}, $j\in R$,
and an allocation $x_i \in X_i$, $i\in N$, of items to players so that the
following conditions are satisfied
(cf.~\cite[Definition~11.12]{Blumrosen_Nisan_2007}):
\begin{enumerate}
\item $\sum_{i\in N} x_i \leq u$ (demand is bounded by the supply),
\item $x_i\in \arg\max_{z_i \in X_i}\{\pi_i(z_i)-\lambda^\top z_i\}\; \forall i\in
  N$ (players optimize their utility),
\item $\lambda^\top(\sum_{i\in N} x_i - u)=0$ (prices clear the market).
\end{enumerate}
While the first two conditions are self-explanatory, the last
complementarity condition is more interesting.
It requires that item prices may only be strictly positive  if these items
are sold at capacity. This condition resembles the economic principle
of a competitive equilibrium, where prices lead to a balance of demand
and supply.\footnote{If this condition is omitted and $0\in
    X_i$ for every player, sufficiently large finite prices support
    the empty allocation as a utility-maximizing choice for
    everyone.}

The \Walras equilibrium concept is quite convincing: given market
prices, the players are happy with their bundle, no envy among players
occurs, the market clears, and the prices are simple, anonymous, and
can be easily communicated.
A further remarkable property is captured in the first welfare
theorem of economics: When a \Walras equilibrium exists, the allocation
maximizes the social welfare defined as the sum of the valuations of
the players.
There is, however, a main drawback, especially for combinatorial
markets, because the existence of \Walras equilibria is not guaranteed.
An equilibrium exists only
for special cases, e.g., when valuations satisfy a gross-substitute (GS)
property \parencite{Kelso82,Gul99,Ausubel2002} or,
more generally, when they satisfy a discrete convexity condition
\parencite{Danilov2001}.

\subsection{Our Results}

In order to address the possible non-existence of \Walras equilibria,
we propose the concept of \emph{regret \Walras equilibria} as a
relaxation. Here, we search for a tuple~$(x,\lambda)$, i.e.,
an allocation $x=(x_i)_{i\in N} \in X \define \times_{i\in N} X_i$
respecting the supply bounds and market-clearing prices $\lambda\in
\R_+^m$ that together minimize the induced regret defined as
\[ \Reg(x,\lambda) \define \sum_{i\in N}\left(\max_{z_i\in X_i}
    \Set{\pi_i(z_i)-\lambda^\top z_i}- \left(\pi_i(x_i)-\lambda^\top
    x_i\right)\right).\]
Total regret is the sum of the differences between each
  player's best-response utility at $\lambda$ and the utility of its
  allocated bundle.
For a capacity-feasible allocation and complementary
  nonnegative prices, we call $(x,\lambda)$ a $\Delta$-regret \Walras
  equilibrium if $\Reg(x,\lambda)\leq\Delta$, where $\Delta\geq0$.
  Hence, zero regret is equivalent to an exact \Walras
  equilibrium. At fixed prices, each player solves an independent
  utility-maximization problem; total regret at most $\Delta$ implies
  that each player's utility is within $\Delta$ of its
  optimum.
Note that the term ``regret'' is different to that in the area
of \emph{regret learning} in games; see also the recent works
by \textcite{DASKALAKIS2022,Branzei23} in combinatorial auctions. There, regret is
defined as the average utility gained over the history of play
compared to a best fixed strategy in hindsight.

\subsubsection*{Characterization Results}

The central idea is to design an \emph{allocation-dependent} optimization problem
whose dual prices satisfy market clearing at a prescribed allocation~$x$.
The identities relating regret to Lagrangian gaps and relating the dual to
a convexified primal are classical; see~\textcite{GHP07,HB17,Geoffrion74}.
Our construction selects the capacity constraints as a function of~$x$,
thereby enforcing the additional requirement that prices of unsold items
remain exactly zero even when buyer regret is positive.
This yields a complete characterization of achievable regret through an
$x$-\emph{parameterized} social welfare problem
\[ \max_{z \in X} \quad \pi(z) \define \sum_{i\in N} \pi_i(z_i)
  \quad \text{s.t.} \quad
  \ell_j(z) \leq \ell_j(x)\;  \forall j\in R \text{ with } u_j=\ell_j(x), \]
where for $z\in X$ we use the shorthand $\ell(z) \define \sum_{i\in N}z_i$.
Here, $x$ is a fixed capacity-feasible allocation, whereas $z$ is the
optimization variable.
Let us associate with this problem the dual problem
\[ \min\defset{\mu_x(\lambda)}{\lambda\in \R^m_+ \text{ with
    }\lambda_j=0 \text{ for all }j\in R \text{ with }u_j>\ell_j(x)},\]
where $\mu_x(\lambda)$ is the $x$-parameterized Lagrangian dual function. The  $x$-parameterized \emph{duality gap}
of a primal-dual feasible pair $(z,\lambda)$ is defined as
$\mu_x(\lambda)-\pi(z)$.
Within this framework, we establish the following characterization result
without further structural assumptions on valuations or the finite strategy spaces:
 
 \clearpage
\begin{rsttheorem}{\ref{thm:main-general}}[and Corollaries~\ref{cor:main-general} and~\ref{cor:nec} --  Informal]\label{thm:intro-nec}
  \hfill
  \begin{enumerate}
  \item 	The regret of a capacity-feasible $x$ \wrt a complementary price vector $\lambda$ equals its
    $x$-parameterized duality gap $\mu_x(\lambda)-\pi(x)$
    and is minimized for fixed $x$ precisely when $\lambda$ is optimal for the $x$-parameterized dual problem.
  \end{enumerate}
  The convexification of the $x$-parameterized social welfare problem leads to a certain
  $x$-parameterized configuration~LP and we show that its integrality gap \wrt $x$ coincides with the $x$-parameterized duality gap, implying with (a):
  \begin{enumerate}[resume]
  \item For $\Delta\geq0$, there exists a $\Delta$-regret \Walras equilibrium if and only if there exists $x$ with $\ell(x)\leq u$ and the corresponding $x$-parameterized integrality gap
    of the configuration LP being upper bounded by $\Delta$.
  \end{enumerate}
  Finally, by observing that the $x$-parameterized integrality gap is an upper bound on the integrality gap of the classical social welfare problem, we obtain the following necessary condition by (a):
  \begin{enumerate}[resume]
  \item If $(x,\lambda)$ is a $\Delta$-regret \Walras equilibrium, then the integrality gap of the classical social welfare problem is bounded by $\Delta$.
  \end{enumerate}
\end{rsttheorem}

As a direct consequence of this general result, we get for
 the special case of weakly monotone valuations and upwards-closed strategy spaces:

\begin{rstcorollary}{\ref{cor:monotone}}[Informal]
  Consider instances with weakly monotone valuations and upwards-closed strategy spaces, and allocations $x\in X$
  for which, w.l.o.g., the entire capacity $u$ is used, i.e. $\ell(x)=u$.
  \begin{enumerate}
  \item For any such pair $(x,\lambda)$, its actual regret $\Reg(x,\lambda)$ equals its duality gap for the classical social welfare problem.
  \item The optimal $\Delta$-regret \Walras equilibrium is obtained by a pair of primal-dual optimal solutions of the social welfare problem.
\end{enumerate}
\end{rstcorollary}
This last result implies that (for weakly-monotone valuations and upwards-closed strategy spaces) the first welfare theorem
remains valid for the concept of regret \Walras equilibria, i.e., the best possible regret for any instance is
achieved for an optimal social welfare solution.
For general strategy spaces and valuations, however,  Proposition~\ref{prop:welfare-regret} shows that minimum regret may
require sacrificing social welfare.

\paragraph*{Buyer and Seller Regret}
Complementarity can also be interpreted from the sellers' perspective.
Suppose that each resource~$j$ is offered by a price-taking seller with
zero cost who owns all $u_j$ units. At prices~$\lambda$, such a seller
wants to sell all units whenever $\lambda_j>0$. If only $\ell_j(x)$ units
are sold at allocation~$x$, the seller forgoes the revenue
$\lambda_j(u_j-\ell_j(x))$, which we call the \emph{seller regret}. The
total seller regret $S(x,\lambda)=\lambda^\top(u-\ell(x))$ vanishes if and
only if~$\lambda$ is complementary to~$x$, i.e., if $\lambda$ belongs to
$$\Lambda(x,u)\define\defset{\lambda\in\R_+^m}{\lambda_j = 0 \text{ for
    all } j \in R \text{ with } \ell_j(x)<u_j}.$$
Instead of imposing complementarity as a hard constraint, one can
therefore penalize seller regret with a weight $\tau\geq0$ and minimize
buyer regret plus $\tau$ times seller regret over all prices
$\lambda\geq0$. To describe the result, let $\rho(q)$ denote the optimal
value of the classical configuration LP if the capacity vector~$u$ is
replaced by~$q$.

\begin{rsttheorem}{\ref{thm:weighted-regret}}[Informal]
For every feasible allocation $x$ and weight $\tau\geq0$, the minimum
of buyer regret plus $\tau$ times seller regret equals the
configuration-LP gap at the adjusted capacity vector
$\ell(x)+\tau(u-\ell(x))$:
\[
 \min_{\lambda\geq0}\{\Reg(x,\lambda)+\tau S(x,\lambda)\}
 =\rho\bigl(\ell(x)+\tau(u-\ell(x))\bigr)-\pi(x).
\]
At $\tau=1$, this recovers the classical welfare-dual gap at $x$.
\end{rsttheorem}

\begin{rsttheorem}{\ref{thm:finite-penalty}}[Informal]
For every feasible allocation $x$, there is an explicit finite threshold
$\tau_0(x)$ such that, for all $\tau\geq\tau_0(x)$,
\[
 \min_{\lambda\geq0} \Set{\Reg(x,\lambda)+\tau S(x,\lambda)}
 =\min_{\lambda\in\Lambda(x,u)}\Reg(x,\lambda).
\]
If $\tau>\tau_0(x)$, every minimizing price vector assigns price zero
to all resources with unused capacity at $x$.
\end{rsttheorem}

In words, penalizing seller regret with a sufficiently large but
finite weight yields exactly the minimum buyer regret under
complementarity, which is the value characterized by our
allocation-dependent program. For $\tau=1$, in contrast, one obtains the
classical duality gap used in convex hull pricing.

\paragraph*{The Cost of Complementarity}
Let $\pi^{\LP}$ and $\pi^{\LP}_x$ denote the optimal values of the
classical and the $x$-parameterized configuration LP, respectively. Since
the parameterized LP drops the capacity constraints of slack resources,
$\pi^{\LP}_x\geq\pi^{\LP}$.
Hence, the minimum regret at $x$ decomposes into two nonnegative terms:
\[
 \min_{\lambda\in\Lambda(x,u)}\Reg(x,\lambda)
 =\bigl(\pi_x^{\LP}-\pi^{\LP}\bigr)
   +\bigl(\pi^{\LP}-\pi(x)\bigr).
\]
The second term is the classical integrality gap of~$x$, which equals
the minimum of buyer plus seller regret ($\tau=1$). The first term is the
additional cost of enforcing complementarity at~$x$; it vanishes if all
capacities bind at~$x$. This decomposition, derived in
Section~\ref{sec:complementarity-cost}, explains why allocations and
prices have to be chosen jointly. An allocation with high welfare may
leave capacity unused and thereby incur a high cost of complementarity,
as in Proposition~\ref{prop:welfare-regret}.

\paragraph*{Price Adjustment}
Finally, we bound the regret that results if arbitrary prices are made
complementary by setting the prices of all slack resources to zero. The
bound is stated in terms of the classical duality gap
$\gamma(x,\lambda)$ of $x$ and $\lambda$, which equals buyer plus seller
regret.

\begin{rsttheorem}{\ref{thm:suf-condition}}[Informal]
For any feasible allocation $x$ and nonnegative prices $\lambda$,
set to zero the prices of resources with slack at $x$, obtaining
$\bar\lambda\in\Lambda(x,u)$. Then,
\[
 \Reg(x,\bar\lambda)\leq(1+\eta(x))\gamma(x,\lambda)
 \leq\bigl(1+(n-1)u_{\max}\bigr)\gamma(x,\lambda)
\]
holds, where
\[
 \eta(x)=\max_{j:\ell_j(x)<u_j}
 \frac{(M_j-u_j)_+}{u_j-\ell_j(x)},\quad
 M_j=\sum_{i\in N}\max_{z_i\in X_i}z_{ij},
\]
with maximum zero when all capacities bind,
$(a)_+=\max\{a,0\}$, and $u_{\max}=\max_{j\in R}u_j$.
\end{rsttheorem}
The uniform factor uses that slacks are integral and that every
bundle satisfies $z_{ij}\leq u_j$. By Example~\ref{ex:repair-sharp}, it
cannot be improved as long as the allocation~$x$ is kept fixed; whether
it is also tight if allocations and prices are optimized jointly remains
open.
For monotone valuations and upwards-closed strategy spaces, a standard
vertex-rounding argument also bounds globally minimum regret by the
sum of the $\min\{m,n\}$ largest individual maximum valuations;
see~\eqref{eq:fractional-player-bound}.

\subsubsection*{Polynomial-Time Algorithms}

These structural results allow us to employ existing welfare algorithms
through oracle access. We use an additive adaptation of the multiplicative
LP-gap-verification concept by \textcite[Definition~2.1]{ElbassioniFS13}.
The required guarantee is $\rho(u)-\pi(\ALG(I))\leq\alpha(I)$, measured
against the configuration LP. A multiplicative LP guarantee $\beta\geq1$
gives $\alpha(I)=(1-1/\beta)\rho(u)$.

\begin{rsttheorem}{\ref{thm:PolyAlg}}
  Let $\mathcal{I}$ be a class of instances of the social welfare
  problem satisfying the rational encoding and polynomial-time oracle assumptions of Section~5.
  Let $\ALG$ be a polynomial-time algorithm verifying an additive
  integrality gap of (at most) $\alpha(I)\geq0$ for the social welfare
  problem.
  Then, the following holds true.
  \begin{enumerate}
  \item\label{item:1} If $\mathcal{I}$ contains only instances
    with monotone valuations and upwards-closed strategy spaces, then there is a polynomial-time
    algorithm (based on $\ALG$) that computes $\Delta$-regret \Walras
    equilibria with $\Delta\leq \alpha(I)$.
  \item\label{item:2}  If $\mathcal{I}$ contains general
    instances  (general valuations), then there is a polynomial-time
    algorithm (based on $\ALG$) that computes $\Delta$-regret \Walras
    equilibria  with $\Delta\leq
    \alpha(I)(1+(n-1)u_{\max})$.
  \end{enumerate}
\end{rsttheorem}

\subsubsection*{Lower Bounds}

As  mentioned before, we can use the necessary optimality
conditions of Theorem~\ref{thm:intro-nec} to establish lower bounds on the existence
of low-regret equilibria. While the integrality gap gives an instance-specific
non-existence certificate, we can even employ computational complexity theory
to obtain non-existence of good-regret bounds for classes of valuations.
The following result is greatly inspired by the one by~\textcite{Roughgarden:2015}
relating the existence of exact \Walras equilibria to the complexity dichotomy
of the demand and the social welfare problem.

\begin{rsttheorem}{\ref{thm:ConnectionNPComp}}[Informal]
Consider a class satisfying the computational assumptions of
Section~\ref{sec:algorithms}, and fix a rational $0\leq\varepsilon<1$.
If approximating its optimal welfare value from below within factor
$1-\varepsilon$ is impossible in polynomial time unless $P=NP$, then,
assuming $P\ne NP$, some instance admits no equilibrium with regret
at most $\varepsilon\rho(u)$. An analogous statement holds for a fixed
additive approximation error~$\Delta$.
\end{rsttheorem}
The hardness premise concerns a numerical estimate of the optimal
value, rather than only the search for an allocation. The relative
statement is invariant under scaling all valuations; Section~\ref{sec:complexity}
also explains the limitation of fixed additive thresholds on classes
closed under scaling.

\subsubsection*{Applications}
Our results relate to pricing in nonconvex electricity markets, where
total regret corresponds to the lost opportunity costs of the
buyers~\parencite{GHP07,HB17}. Two practices can be distinguished there.
Pricing rules such as convex hull pricing price a given welfare-maximizing
dispatch~\parencite{GHP07,HB17,Andrianesis22,guo2021copositive,AhunbayBK25}.
For any prescribed feasible allocation, Theorem~\ref{thm:main-general}
yields the complementary prices with minimum regret together with the exact
value of this regret, and Section~\ref{sec:regret-bounds} relates this
value to the convex hull pricing objective. European day-ahead markets, in
contrast, determine the allocation jointly with uniform prices: welfare is
maximized subject to the condition that no accepted order incurs a loss at
these prices, so profitable orders may be rejected and welfare may be
lost~\parencite{MartinMP14}. Proposition~\ref{prop:welfare-regret} shows a
similar effect in our model: under exact market clearing, a
welfare-maximizing allocation can have strictly larger minimum regret than
another allocation, so it may be beneficial to choose allocations and prices.
Our model is single-sided, with fixed capacities and zero-cost
sellers. It does not model nonconvex generation or a particular
market-clearing rule, and extending our construction to two-sided markets
requires a separate analysis.

\subsection{Related Work}
\subsubsection{Existence of \Walras Equilibria}
The existence of \Walras equilibria and their computation
is a central topic in several areas and consequently, there
  is quite a large literature.
Let us refer here to the
survey by \textcite{Bichler21} for a comprehensive overview.
For the problem of allocating indivisible single-unit items,
there are several characterizations of the
existence of equilibria related to
the gross-substitute property of valuations;
see~\textcite{Kelso82,Gul99,Ausubel2002}.
Several works established connections
of the equilibrium existence problem w.r.t.\ LP-duality and
integrality \parencite{Bikhchandani1997,Shapley1971}.
\textcite{MurotaT03,Murota:2003} established connections
between the gross substitutability property and M-convexity properties
of demand sets and valuations.

\textcite{Danilov2001} investigated the existence of \Walras equilibria
for multi-unit auctions and identified general
conditions on the demand sets and valuations
related to discrete convexity; see also
\textcite{Milgrom2009,Ausubel2006,FujishigeY03,Sun2009}.
\textcite{Baldwin2019} explored a connection with tropical
geometry and gave a necessary and sufficient condition for the existence
of a competitive equilibrium in product-mix auctions of indivisible
goods.
For a comparison of the above works, especially with respect to the
role of discrete convexity, we refer to the excellent survey
by \textcite{Shioura2015}.
\textcite{Candogan2018,CandoganP18}  show that valuation classes
(beyond GS valuations) based on graphical structures also imply the
existence of \Walras equilibria. Their proofs use integrality of
optimal solutions of an associated linear min-cost flow
formulation. LP characterizations for the existence of \Walras
equilibria were given
by~\textcite{Bikhchandani1997,BikhchandaniO02,Candogan2018,Roughgarden:2015}.

\textcite{HarksSchwarzPricing} develop a unified
pricing framework for nonconvex resource allocation games based on
Lagrangian duality and convexification, covering \Walras equilibria
as one application. Their notion of weak enforceability with
market-clearing prices requires exact individual optimality and zero
prices on resources with slack capacity. We retain the latter price
restriction while allowing positive buyer regret; our allocation-dependent
welfare programs characterize and minimize this regret at a prescribed
allocation.

\subsubsection{Convex Hull Pricing}
Convex hull pricing relates aggregate lost opportunity costs to a
Lagrangian duality gap; see~\textcite{GHP07,HB17}. The equality between
the Lagrangian dual value and the value obtained by convexifying the
individual feasible sets and objective functions is also classical;
see~\textcite{Geoffrion74} and \textcite[Theorem~1]{HB17}.
The distinguishing construction here
is the \emph{allocation-dependent} choice of retained capacity constraints,
which imposes complementarity at the prescribed integral allocation.
Complementarity at an optimal fractional configuration need not imply
complementarity at that integral allocation.
For example, the basic electricity formulation by \textcite[Section~2]{HB17}
dualizes a supply-demand balance equality, with unrestricted balance
multipliers. Our capacity inequalities have nonnegative multipliers, and
the allocation selects a face of that nonnegative price domain.
The difference concerns the precise primal constraints and admissible
prices, rather than the duality identity.
Other modifications of convex hull pricing concern different program
components. For example, \textcite{Borokhov17} modifies individual
feasible sets used in the opportunity-cost calculation, whereas our
construction keeps those sets fixed and selects coupling inequalities.
The pricing framework by \textcite{AhunbayBK25} starts from a
welfare-maximizing allocation and compares several pricing objectives.
Here, the allocation can be prescribed arbitrarily or optimized jointly
with complementary prices.
The bound-form first welfare theorem of \textcite{MilgromWatt26}
relates welfare loss to rationing losses and priced excess supply.
Section~\ref{sec:price-repair} explains its fixed-supply counterpart in
our model. Our characterization additionally minimizes regret under
complementarity at a prescribed allocation. The weighted-seller family
and finite-penalty result connect this requirement to equal weighting
of opportunity costs.

\subsubsection{Relaxations of \Walras Equilibria}
There have been several proposals of relaxed notions of \Walras equilibria in the literature in order
to recover existence.
One way is to allow arbitrary bundle prices instead of item prices,
see \textcite{BikhchandaniO02,Roughgarden:2015}. This approach leads
to stronger existence results but loses the simple structure of
item prices. Other works relax market clearing;
see~\textcite{Budish11,Guruswami05}.
A particularly relevant comparison is the work of \textcite{FeldmanGL15}. They retain
exact buyer demand optimality while allowing positive prices on unsold
items. Their Corollary~1 characterizes stable outcomes selling an item
set $M'$ by a unit integrality gap of the configuration LP restricted to $M'$.
Our construction instead keeps every individual bundle available and
removes the capacity constraints of unsold items, forcing their dual
prices to zero while allowing positive buyer regret. Thus, the selected
item set and the selected capacity constraints give different LP families.
The poor welfare of stable item pricing motivates bundling; allowing
bundles as priced objects can recover at least half of optimal welfare
for arbitrary valuations~\parencite{FeldmanGL16}.
Our paper proposes to stick with market clearing but relax optimality
of players' strategies---measured in terms of total regret, which is an
\emph{additive} form of utility approximation.
Multiplicative notions of approximate market equilibria have been
considered by~\textcite{Codenotti05,GargTV25} for special (concave)
valuations and the divisible good setting.
Equilibrium computation and inapproximability in divisible
Fisher markets with separable piecewise-linear concave utilities are
studied by \textcite{Vazirani11,Deligkas24}; these are distinct from the
relaxation of zero prices on unsold indivisible items considered above.
There have been numerous works using additive approximations of (Nash)
equilibria; see~\textcite{DeligkasFS20} and references therein. Let
us refer to \textcite{Daskalakis13}
for a detailed overview on pros and cons of additive versus multiplicative
approximations of equilibria and also how one can be converted to the other.


%% file: model-v2.tex
\section{Model}\label{sec:model}

A combinatorial allocation model is described by a tuple
$I=(N,R,u,X,\pi)$, where $N=\{1,\dots,n\}$ describes a nonempty and finite
set of players and $R=\{1,\dots, m\}$ denotes a nonempty and finite set of
items or resources that are available with multiplicity $u_j\in \Z_+,
j\in R$.\label{sym:players}
Here and in what follows, $\Z_+$ and $\R_+$  denote the nonnegative natural
and real numbers including $0$, respectively.
The set $X \define \times_{i\in N} X_i$ describes the combined strategy space
of the players, where $X_i=\{x_i^1,\dots,x_i^{k_i}\} \subseteq
\Z_+^m$ with $x_i^j\neq x_i^l$ for $j\neq l$  is the nonempty and finite integral strategy space of
player~\mbox{$i \in N$}. We assume $0\in X_i$ for every player.
We have $k_i \define |X_i|\in \N$ and define $k \define \sum_{i\in N}k_i$.\label{sym:strategies}
For $ x_i=(x_{ij})_{j\in R} \in X_i$, the entry $x_{ij}\in \Z_+$ is
the integer amount of resource $j$ consumed by player $i$.
We  call the vector of resource usage $ x=(x_{ij})\in X\subset
\Z_+^{n m}$ a \emph{strategy profile}.
Given $ x\in X$, we can define the \emph{load} on resource $j\in R$ as
$\ell_j( x) \define \sum_{i\in N}  x_{ij},$ where $x_{ij}$ is the
$j$th component of $x_i$.\label{sym:load}
We assume that every $x_i\in X_i$, $i\in N$, is capacity feasible, meaning
that $x_{ij}\leq u_j$ holds for all $j\in R, i\in N$. This does not imply
that every $x\in X$ is capacity feasible, i.e., $\ell(x)\leq u$ does
not need to hold for all $x\in X$.
Let us define by $X(u) \define \defset{x\in X }{\ell(x)\leq u}$ the
set of capacity feasible strategy profiles.\label{sym:Xu}
An important special case arises, if $X_i\subset \{0,1\}^m , i\in N$.
In this case, there is a one-to-one correspondence between $x_i\in
X_i$ and the subset $S_i \define \defset{j\in R}{x_{ij}=1} \subseteq R$
and, thus, we can use the notation~$S_i$ and~$x_i$ interchangeably.

We assume that the utility or valuation function of a player $i \in N$
maps the obtained resources $x_i \in X_i$ to some utility value
$\pi_i(x_i)\in \R$ for some function  $\pi_i: X_i\to \R$
that satisfies $\pi_i(0)=0$ for all $i\in N$.
For a strategy profile $x\in X$, we denote by
  $\pi(x)\define\sum_{i\in N}\pi_i(x_i)$ the resulting \emph{social welfare}.\label{sym:welfare}
The finite sets $X_i$ allow structured bundle domains,
  including the capacity-bounded integral flows in Section~\ref{sec:complexity}.
  Generalized network utilities for divisible goods, as studied by
  \textcite{GargTV25}, provide a related but distinct modeling example.
In the remainder of the paper, unless stated otherwise,
we assume that the resource allocation model is of the form described above.
Throughout the paper, we introduce several notations, which we
  summarize in a list of symbols in Appendix~\ref{app:symbols}.

We are concerned with the problem of defining \emph{item prices}
$\lambda_j \geq 0$, $j\in R$, on the resources in order to clear the
market.
If player $i$ uses item $j$  at level $x_{ij}$, she needs to pay
$\lambda_j  x_{ij}$. The quantities $\pi_i(x_i)$ and
$\lambda^\top  x_i$ are assumed to be normalized
to represent the same unit and we assume that the overall utility is quasi-linear:
$ \pi_i(x_i)-\lambda^\top  x_i$.
\begin{definition}[\Walras Equilibria]\label{def:walras}
  A tuple $( x^*, \lambda^*)\in
  X\times\R_+^m$  is a \emph{\Walras equilibrium}, if the
  following two conditions are satisfied:
  \begin{enumerate}[ref=Definition~\ref{def:walras}(\alph*)]
  \item\label{cond1-w} $\pi_{i}(x^*_i)-(\lambda^*)^{\top}
    x_i^*\geq  \max \defset{\pi_{i}(x_i)-(\lambda^*)^{\top}
      x_i}{x_i\in X_i}$ for all $i\in N$.
  \item\label{cond2w} $x^*\in X(u)$ and $\lambda^*\in \Lambda(x^*,u)
    \define \defset{\lambda\in \R^m_+}{ \lambda^\top(\ell(x^*)-u)=0}$.
  \end{enumerate}
\end{definition}
\ref{cond1-w} requires that every player is happy with the
current item bundle~$x^*_i$ given the prices $\lambda^*$.
\ref{cond2w} requires $x^*$ to be capacity feasible and
$\lambda^*\in \Lambda(x^*,u)$ to fulfill the economic principle that
resources with slack demand must have zero price, i.e.,
$\ell_j( x^*) < u_j$ implies $\lambda_j^*=0$ for all $j\in R$.

Note that \Walras equilibria need not exist due to the integrality
(non-convexity) of the strategy spaces and the possible non-convexity
of the valuation functions.
In this paper, we propose a simple but natural relaxation of \Walras
equilibria towards the notion of \emph{regret \Walras equilibria}.
To this end, we define for $( x, \lambda)\in X\times\R_+^m$ the
\emph{regret} of player $i\in N$ as
\begin{equation*}
  \Reg_i( x_i, \lambda) \define \max
  \Defset{\pi_{i}(z_i)-\lambda^{\top} z_i}{z_i\in X_i}
  - \left(\pi_{i}(x_i) - \lambda^{\top} x_i\right).
\end{equation*}
The regret measures the distance from the current utility value under
$( x_i, \lambda)$ to the maximum utility value achievable for
feasible bundles $z_i\in X_i$ under prices $\lambda$.
Note that $\Reg_i( x_i, \lambda)\geq 0$.
The \emph{total regret} for a tuple  $( x, \lambda)\in X\times \R^m_+$
is defined as
\begin{equation*}
  \label{sym:regret}
  \Reg(x, \lambda) \define \sum_{i\in N}\Reg_i( x_i, \lambda).
\end{equation*}
For fixed $\lambda$, the regret function is equal to the
Nikaido--Isoda function of the induced strategic game; see
\textcite{nikaido1955} and \textcite{Harks24}.
For $x^*\in X(u)$ and $\lambda^*\in\Lambda(x^*,u)$, we
  obtain the following equivalences%
  \begin{align*}
    (x^*,\lambda^*)\text{ is a \Walras equilibrium}
    &\iff \Reg_i(x_i^*,\lambda^*)=0\quad\forall i\in N\\
    &\iff \Reg(x^*,\lambda^*)=0.
  \end{align*}
\begin{definition}[$\Delta$-Regret \Walras Equilibria]
  \label{def:apx-implementable}
  Let $\Delta\in \R_+$.
  A tuple $( x^*, \lambda^*)\in
  X\times\R_+^m$ is a \emph{$\Delta$-regret \Walras
    equilibrium}, if the following two conditions are satisfied:
  \begin{enumerate}[ref=Definition~\ref{def:apx-implementable}(\alph*)]
  \item\label{cond2} $ \Reg( x^*, \lambda^*) \leq \Delta$.
  \item\label{cond3} $x^* \in X(u)$ and $\lambda^*\in \Lambda(x^*,u)$.
  \end{enumerate}
\end{definition}
\ref{cond2} relaxes the condition in~\ref{cond1-w}
in the sense that a \Walras equilibrium is a special $\Delta$-regret
\Walras equilibrium for $\Delta=0$.
\ref{cond3} is the same as~\ref{cond2w}.
A $\Delta$-regret equilibrium $(x^*,\lambda^*)$
is also an \emph{additive} approximate \Walras equilibrium (in the
sense that \mbox{$\Reg_i(x_i^*,\lambda^*)\leq\Delta$} for all $i\in N$).
Conversely, an additive $\Delta$-approximate \Walras equilibrium
induces an $n\Delta$-regret equilibrium.

Beyond the finite, individually capacity-bounded domains and the normalization at the empty bundle, we impose no structural assumptions on valuations or strategy spaces. Exact ($0$-regret) \Walras equilibria therefore need not exist.
Let us give a simple example of an instance that
has no exact \Walras equilibrium but
admits a $0.5$-regret \Walras equilibrium.

\begin{example}
  \label{ex:nonexistence}
  This is the standard triangle construction; see \textcite[Section~5.2]{HKMN11}.
  There are three players $N=\{1,2,3\}$ and three resources
  $R=\{1,2,3\}$ with unit capacity each $u=(1,1,1)$.
  The strategy spaces are given by $X_i=\{0,1\}^3$ for $i\in N$.
  The players have single-minded (monotone) utility functions of the
  form
  \begin{align*}
    \pi_1(S)
    & = \begin{cases}1, &\text{ if }S\supseteq\{1,2\},\\
      0,& \text{ else,}
    \end{cases}
    \\
    \pi_2(S)
    & = \begin{cases}
      1, &\text{ if }S\supseteq\{2,3\},\\
      0,& \text{ else,}\end{cases}
    \\
    \pi_3(S)
    & = \begin{cases}
      1, &\text{ if }S\supseteq\{1,3\},\\
      0,& \text{ else.}
    \end{cases}
  \end{align*}
  The critical sets are visualized in Figure~\ref{fig:counter-ex}.
    We prove that the globally minimum regret is $1/2$. At most one player can receive its desired pair, so every
    feasible allocation has welfare at most~1. For any prices $\lambda\geq0$,
    let $P\define\lambda_1+\lambda_2+\lambda_3$. Each player can buy its
      desired pair $\{j,k\}$ or nothing, so its best-response utility is at least
      $\max\{1-\lambda_j-\lambda_k,0\}$. Every item lies in exactly two of the
      three desired pairs, so the best-response utilities sum to at least
      $\max\{3-2P,0\}$. For a feasible allocation~$x$ and complementary prices,
      the allocated bundles have total utility
      $\pi(x)-\lambda^\top\ell(x)=\pi(x)-\lambda^\top u=\pi(x)-P$. Hence,
      \[
        \Reg(x,\lambda)\geq\max\{3-2P,0\}-\bigl(\pi(x)-P\bigr)
        =\max\{3-P,P\}-\pi(x)\geq3/2-\pi(x)\geq1/2.
      \]
    The allocation $x_1=\{1,2\}$, $x_2=\{3\}$, $x_3=\varnothing$ and
    prices $\lambda=(1/2,1/2,1/2)$ attain the bound: the individual regrets
    are $(0,1/2,0)$. All capacity is allocated, so the prices are complementary.
    Hence, no exact equilibrium exists and the minimum regret is~$1/2$.
  \begin{figure}
    \centering
    \input{tikz-imgs/example-1.tikz}
    \caption{Construction of a model $I$ with $3$ players and $3$
      resources.}
    \label{fig:counter-ex}
  \end{figure}
\end{example}

The goal of this paper is to understand the optimization problem
\begin{align}
  \tag{$\mathrm{Min}$-$\mathrm{Regret}$}
  \label{regret-opt}
  \min_{(x,\lambda)} \quad
  & \Reg( x, \lambda)
  \\
  \st \quad
  & \ell(x)\leq u,\label{eq:capacity}\\
  & \lambda^\top(\ell(x)- u)=0, \label{eq:complementarity}\\
  &(x,\lambda)\in X\times\R^m_+.\notag
\end{align}

Before we study regret \Walras equilibria in more detail, let us
show by a small example that for optimal regret \Walras equilibria
(which always exist in our model) the ``first-welfare-theorem'' property does
not hold anymore. Recall that  whenever a  $0$-regret \Walras equilibrium exists, it is attained at  an optimal solution
of the classical welfare-maximization problem (cf.\ \cite[Proposition~1]{Bikhchandani1997}):
\begin{align}
  \max_{x} \quad \pi(x)
  \quad \st \quad
  \ell(x) \leq u,\; x \in X.
  \tag{P}\label{price-opt}
\end{align}
\begin{proposition}\label{prop:welfare-regret}
  There are combinatorial markets $I$ such that for any optimal
  solution $(x^*,\lambda^*)$ of \eqref{regret-opt},
  $x^*$ is not optimal for the social welfare problem \eqref{price-opt}.
\end{proposition}
\begin{proof}
  Consider a model~$I$ with $4$ players and $7$ resources,
  each with a capacity of~$1$ as illustrated in Figure~\ref{fig:ex1}.
  Player $1$ has two non-trivial  strategies (except the empty set)
  depicted by the two blue horizontal sets $x_1$ and $x'_1$.
  The other three players have each one non-trivial strategy
  $x_2, x_3,x_4$; see Figure~\ref{fig:ex1}.
  The valuations for non-trivial strategies are given as
  \[
    \pi_1(x_1)=1.5,
    \quad \pi_1(x'_1)=1,
    \quad \pi_2(x_2)=1,
    \quad \pi_3(x_3)=1,
    \quad \pi_4(x_4)=1.
  \]
  Clearly, the unique optimal solution to~\eqref{price-opt}
  is $x^*=(x_1,0,0,0)$ with a value of~$1.5$.
  Note that any pair of nontrivial strategies of different players
  overlaps for at least one resource.
  Starting with $x^*$, we can only define positive prices on the
  resources~$5,6,7$.
    Let $p_5,p_6,p_7\geq0$ be arbitrary complementary prices and set
    $P=p_5+p_6+p_7$. All other prices are zero. The regret of player~1 is
    $\max\{P-1/2,0\}$, and the regrets of the other players are
    $\max\{1-p_5,0\}$, $\max\{1-p_6,0\}$ and $\max\{1-p_7,0\}$, respectively.
    Consequently,
    \[
      \Reg(x^*,p)\geq(P-1/2)+(3-P)=5/2.
    \]
    The choice $p_5=p_6=p_7=1/4$ attains this bound. Thus, the minimum
    regret at the welfare-optimal allocation is exactly~$5/2$.
    Instead, allocate $x_2=\{1,2,5\}$ to player~2 and nothing to the others,
    so $\tilde x=(0,x_2,0,0)$ and $\pi(\tilde x)=1$.
    Set $\lambda_1=1$, $\lambda_5=1/2$, and all other prices to zero.
    The individual regrets are
    \[
      \Reg_1(\tilde x_1,\lambda)=1,\qquad
      \Reg_2(\tilde x_2,\lambda)=1/2,\qquad
      \Reg_3(\tilde x_3,\lambda)=\Reg_4(\tilde x_4,\lambda)=0.
    \]
  Thus, total regret is $3/2<5/2$, proving the proposition.
  In fact, $3/2$ is the globally minimum regret. Every nonempty feasible
  allocation assigns exactly one nontrivial bundle, since any two bundles
  of different players intersect. At $\tilde x$, arbitrary complementary
  prices $p_1,p_2,p_5\geq0$ give
  \begin{align*}
    \Reg(\tilde x,p)
    &\geq (3/2-p_5)+(p_1+p_2+p_5-1)+2\max\{1-p_1,0\}\\
    &=1/2+p_2+p_1+2\max\{1-p_1,0\}\geq3/2.
  \end{align*}
  The same calculation applies when player~3 or~4 receives its bundle.
  The welfare-optimal allocation has minimum regret $5/2$, as proved
  above. If only $x'_1=\{2,3,4\}$ is assigned, then, for arbitrary
    complementary prices $p_2,p_3,p_4\geq0$, summing player~1's
  deviation to $x_1$ and the other players' deviations to their nontrivial
  bundles gives
  \[
      \Reg\bigl((x'_1,0,0,0),p\bigr)
      \geq\bigl(3/2-(1-p_2-p_3-p_4)\bigr)+(1-p_2)+(1-p_3)+(1-p_4)=7/2.
    \]%
  Finally, the empty allocation
  forces all prices to zero and has regret $9/2$.
  \begin{figure}
    \centering
    \input{tikz-imgs/example-2.tikz}
    \caption{Construction of a model $I$ with $4$ players and $7$
      resources.}
    \label{fig:ex1}
  \end{figure}
\end{proof}



%% file: tikz-imgs/example-1.tikz
\scalebox{0.75}{
  \begin{tikzpicture}[scale=1,every
    node/.style={draw,circle,minimum size=7mm,inner
      sep=1pt,thick},
    line/.style={thick}]

    \draw[line,rounded corners, fill=blue!20, opacity=0.5] (-1.5,0.5) rectangle (2.5,-0.5);
    \draw[line,rounded corners, fill=red!20, opacity=0.5] (-2,-0.5) --(-0.75,-0.5)--(1.5,2)-- (0.25,2) --cycle;
    \draw[line,rounded corners, fill=green!20, opacity=0.5] (3,-0.5)-- (1.75,-0.5) -- (-0.5,2)--(0.75,2) --cycle;

    \node (n1) at (0.5,1.5) {1};
    \node (n2) at (-1,0) {2};
    \node (n3) at (2,0)  {3};
  \end{tikzpicture}
}


%% file: tikz-imgs/example-2.tikz
\scalebox{0.75}{
  \begin{tikzpicture}[scale=1,every node/.style={draw,circle,minimum size=7mm,inner sep=1pt,thick},
    line/.style={thick}]
    \draw[line,rounded corners, fill=blue!20, opacity=0.5] (-2.7,1.5) rectangle (2.7,0.5);
    \draw[line,rounded corners, fill=blue!20, opacity=0.5] (-3.7,-0.5) rectangle (3.7,-1.5);
    \draw[line,rounded corners, fill=orange!20, opacity=0.5] (-0.1,3.5) -- (1,3.5) -- (-3,-2) -- (-4,-2) -- cycle;
    \draw[line,rounded corners, fill=green!20, opacity=0.5] (-0.5,3.5) -- (0.5,3.5) -- (0.5,-2) -- (-0.5,-2)-- cycle;
    \draw[line,rounded corners, fill=black!20, opacity=0.5] (-0.9,3.5) --(0.1,3.5)-- (4.1,-2)--(3.1,-2)  --cycle;
    \node[draw=none] (z) at (3.75,-2.25)  {$x_4$};
    \node[draw=none] (z) at (-3.75,-2.25)  {$x_2$};
    \node (n1) at (0,3) {1};
    \node (n2) at (-1.3,1) {2};
    \node (n3) at (0,1)  {3};
    \node (n4) at (1.4,1)  {4};
    \node (n5) at (-2.75,-1) {5};
    \node (n6) at (0,-1)  {6};
    \node (n7) at (2.85,-1)  {7};
    \node[below=5mm of n6,draw=none] {$x_3$};
    \node[right=7mm of n7,draw=none] {$x_1$};
    \node[right=10mm of n4,draw=none] {$x'_1$};
  \end{tikzpicture}
}


%% file: general-valuations-v2.tex
\section{Characterizing Regret Equilibria}\label{sec:GenVal}
  Our central construction chooses the welfare problem as a function of the
  allocation to be priced. For a prescribed $x\in X(u)$, market clearing
  requires zero prices on each resource with slack at~$x$. We therefore
  retain only the capacity constraints that bind at~$x$, so the dual
  automatically has the required price domain. The regret accounting
  identity and the convexification-duality principle used below are
  classical~\parencite{GHP07,HB17,Geoffrion74}; the allocation-dependent
  choice of constraints enforces complementarity while buyer regret is
  allowed to be positive. Define the resulting welfare problem as
\begin{align}
  \tag{$\mathrm{P}(x)$}\label{price-opt-restricted}
  \max_{z} \quad & \pi(z) \\
  \st \quad & \ell_j( z) \leq \ell_j( x)  \text{ for all }j\in R^=(x) ,  \notag\\
                 & z \in X, \nonumber
\end{align}
where $R^=(x) \define \Set{j \in R: \ell_j(x) = u_j}$.
Note that, when comparing the parameterized and classical social welfare problems \eqref{price-opt-restricted} and \eqref{price-opt},
the parameterized problem drops the inequality constraints for all
resources with slack, i.e., for all $j \in   \Rleq( x) \define
\Defset{j\in R}{\ell_j(x)<u_j}$.\label{sym:Req}
In this regard, the right-hand side constraints of~\eqref{price-opt-restricted} are
tailored to obtain market-clearing duals \wrt $x$.
Indeed,
the $x$-parameterized Lagrangian function for~\eqref{price-opt-restricted} is given by
\begin{equation*}
  L_{x}(z,\lambda) \define \pi(z) -\lambda^\top (\ell( z)- \ell(x)) \quad \text{ for }
  \lambda\in \Lambda(x, u).
\end{equation*}
Note that $ \lambda\in \Lambda(x, u)$ implies $\lambda_j=0$ for all
$j\in \Rleq( x)$.
We can define the $x$-parameterized Lagrangian dual function $\mu_{x} : \Lambda(x, u) \to \R$ via
\begin{equation*}
  \mu_{x}( \lambda) \define \max_{ z \in X} \, L_{x}( z, \lambda)
  = \max \Defset{\pi(z)-\lambda^\top (\ell( z)-  \ell(x))}{z \in X}.
\end{equation*}\label{sym:mux}
The $x$-parameterized \emph{dual problem} is then defined as
\begin{align}\label{price-dual-min-restricted}
  \tag{$\mathrm{D}(x)$}
  \min_{\lambda\in  \Lambda(x, u)} \ \mu_{x}(\lambda).
\end{align}
\begin{definition}[Parameterized Duality Gap Function]
  The $x$-parameterized duality \emph{gap function} is defined as
  \begin{equation}
    \label{eq:gapfunction-restricted}
    \gamma_{x}: X(x)\times  \Lambda(x, u)\to \R_+,
    \quad
    (z,\lambda)\mapsto \gamma_{x}(z,\lambda)
    \define \mu_{x}(\lambda) - \pi(z),
  \end{equation}
  where $X(x)$ denotes the feasible set of~\eqref{price-opt-restricted}.
  Note that by weak-duality, $\gamma_x$ is always non-negative.
  We say that \eqref{price-opt-restricted} and \eqref{price-dual-min-restricted} have a
  \emph{duality gap} of $\gamma_{x}^* \define \gamma_{x}(z^*(x),\lambda^*)$,
  where~$(z^*(x),\lambda^*)$ is a primal-dual optimal solution
  to~\eqref{price-opt-restricted} and~\eqref{price-dual-min-restricted}.
\end{definition}
We define the configuration-LP relaxation
  of~\eqref{price-opt-restricted} by letting each player choose
    a convex combination of its bundles (a \emph{configuration})
    instead of a single bundle and then imposing the retained
  coupling constraints on the resulting fractional loads, i.e.,
\begin{align}\tag{$\LP(x)$}\label{LP:price-opt-tilde-restricted}
  \max_{\alpha }\quad
  &\sum_{i\in N} \pi_{i}^\top  \alpha_i
  \\
  \label{ineq:u-convex-finite-tilde-restricted} \text{s.t.} \quad
  &\ell_j(\alpha)\leq  \ell_j(x) \quad  \text{for all }j\in R^=(x),\notag\\
  &\alpha_i \in \Simp_i, \quad i \in N, \nonumber
\end{align}
where $\pi_{i} \define (\pi_i({{z}}_i^l))_{l\in \{1,\dots,k_i\}}$,
$\ell(\alpha) \define \sum_{i\in N}  \sum_{l\in \{1,\dots,k_i\}}
\alpha_{i,l} z_{i}^l$, and $\Simp_i \define \defset{\alpha_i\in \R_{\geq
    0}^{k_i}}{ 1^\top \alpha_i=1}$, $i\in N.$
\begin{definition}[Parameterized Integrality Gap]\label{def:IntGap-restricted}
  Let $ \pi^{\LP}_{x}$ denote the optimal value of \eqref{LP:price-opt-tilde-restricted}.
  The $x$-parameterized \emph{integrality gap function} is defined as
  \begin{equation}
    \label{eq:igapfunction-restricted}
    \iota_{x}: X(x)\to \R_+, \quad z \mapsto \iota_{x}(z) \define
    \pi^{\LP}_{x}-\pi(z).
  \end{equation}
  We say that \eqref{LP:price-opt-tilde-restricted} and~\eqref{price-opt-restricted} have an
  \emph{additive integrality gap} of $\iota_{x}^* \define \iota_{x}(z^*(x))$,
  where~$z^*(x)$ is an optimal solution to~\eqref{price-opt-restricted}.
\end{definition}

  The quantity $\iota_x(z)$ is the LP gap relative to the specified allocation
  $z$. Only $\iota_x^*$ is the gap between the two optimal values. In particular,
  $\iota_x(z)=\iota_x^*+\pi(z^*(x))-\pi(z)$.
  The configuration LP need not describe the convex hull of the globally
  feasible integral allocations.
Analogously, we denote by $\gamma$ and $\iota$ the
gap functions corresponding to the (unparameterized) social welfare
problem~\eqref{price-opt}.
For $\iota$, note that the unparameterized configuration-LP
  relaxation equals \eqref{LP:price-opt-tilde-restricted} for any $x$
  with $\ell(x) = u$, see \eqref{LP:price-opt} in
  Section~\ref{sec:regret-bounds} for the formal definition. Moreover,
  we
denote by
$\mu:\R^m_+ \to \R$ with
\begin{align*}
  \mu(\lambda) \define \max_{ z \in X} \, L( z, \lambda)
  = \max \Defset{\pi(z)-\lambda^\top (\ell(z)-  u)}{z \in X}
\end{align*}
the corresponding Lagrangian dual function and the dual problem via
\begin{align}\label{price-dual-min}
  \tag{D}
  \min_{ \lambda\geq 0} \ \mu( \lambda).
\end{align}

We now derive a complete characterization for the existence of
regret \Walras equilibria. For this, we first show that the regret of
$(x,\lambda)$ equals their duality gap for
\eqref{price-opt-restricted} and~\eqref{price-dual-min-restricted}.
This gap is then shown to be lower bounded by
the $x$-parameterized integrality gap $\iota_x(x)$  with equality precisely for dual-optimal $\lambda$.
No further structural restrictions on strategy spaces or
  valuation functions are needed.

\begin{theorem}\label{thm:main-general}
  For any $(x,\lambda)\in X(u)\times  \Lambda(x, u)$, the following
  holds true.
  \begin{align}\label{eq:general-formula-dual}
    \Reg(x,\lambda)&=\gamma_{x}(x,\lambda)\geq \iota_{x}(x).\end{align}
  Moreover, for any $x \in X(u)$,
  $\lambda^* \in \Lambda(x,u)$ minimizes the regret for $x$
  with $\Reg(x,\lambda^*)=\iota_{ x}(x)$ if and only if
  $\lambda^*$ is optimal for the dual~\eqref{price-dual-min-restricted}.
\end{theorem}
Before proving the above theorem, let us
state an immediate consequence:
\begin{corollary}\label{cor:main-general}
  Let $\Delta\geq0$. Then there exists a $\Delta$-regret
  \Walras equilibrium if and only if there exists $x \in X(u)$ with
  its parameterized integrality gap being bounded by $\Delta$, i.e.,
  $\iota_{x}(x)\leq \Delta$.
\end{corollary}

Moreover, by observing that the parameterized integrality gap
$\iota_x(x)$ is an upper bound on the unparameterized integrality gap
$\iota(x)$ (the feasible region of the classical
  configuration LP is contained in that of the parameterized LP, so
  $\pi^{\LP}_x\geq\pi^{\LP}$), we immediately get another consequence
of Theorem~\ref{thm:main-general}.
\begin{corollary}\label{cor:nec}
  If $(x,\lambda)$ is a $\Delta$-regret \Walras equilibrium, then the
  integrality gap of the classical social welfare problem is bounded
  by $\Delta$, i.e., $\iota^*\leq \Delta$.
\end{corollary}

Since $\ell_j(x)=u_j$ for all $j\in\Req(x)$, both the
  parameterized configuration LP~\eqref{LP:price-opt-tilde-restricted}
  and the price domain $\Lambda(x,u)$ depend on $x$ only through its
  binding set $\Req(x)$. Theorem~\ref{thm:main-general} therefore
  immediately implies the following.
  \begin{corollary}\label{cor:same-binding-set}
    Let $x,x'\in X(u)$ with $\Req(x)=\Req(x')$. Then
    \[
      \min_{\lambda\in\Lambda(x,u)}\Reg(x,\lambda)
      -\min_{\lambda\in\Lambda(x',u)}\Reg(x',\lambda)
      =\pi(x')-\pi(x).
    \]
    In particular, among all feasible allocations with the same binding
    set, minimizing the achievable regret is equivalent to maximizing
    welfare.
  \end{corollary}
  \begin{proof}
    By Theorem~\ref{thm:main-general}, the two minima equal
    $\iota_x(x)=\pi^{\LP}_x-\pi(x)$ and
    $\iota_{x'}(x')=\pi^{\LP}_{x'}-\pi(x')$, respectively, and
    $\pi^{\LP}_x=\pi^{\LP}_{x'}$ because both LPs coincide.
  \end{proof}

Let us now come to the proof of Theorem~\ref{thm:main-general}.
\begin{proof}[Proof of Theorem~\ref{thm:main-general}]
  Let $(x,\lambda)\in X(u)\times  \Lambda(x, u)$ be arbitrary.
  We compute
  \begin{align*}
    \Reg(x,\lambda)
    &=\max_{z\in X}
      \left\{\pi(z)-\lambda^{\top} \ell(z)\right\}-\left(\pi(x)-\lambda^{\top} \ell(x)\right)\\
    &=\max_{z\in X}
      \left\{\pi(z)-\lambda^{\top}\big( \ell(z)- \ell(x)\big)\right\}-\pi(x)\\
    &=\mu_x(\lambda)-\pi(x)=\gamma_{x}(x,\lambda).
  \end{align*}
  To show the inequality in
  \eqref{eq:general-formula-dual}, let us denote by
  $\mu^{\LP}_x(\cdot)$
  the LP-dual of~\eqref{LP:price-opt-tilde-restricted}.
  For fixed~$\lambda$, maximizing each player's affine
    bundle-valuation objective over its simplex gives the same value as
    maximizing over its bundles. Hence $\mu_x(\lambda)=\mu_x^{\LP}(\lambda)$,
    as in the classical convexification principle of
    \textcite[Theorem~1]{Geoffrion74}.
  Moreover, let $\lambda^*\in \Lambda(x,u)$ be optimal
  for~\eqref{price-dual-min-restricted}.
  Then, we get
  \begin{align}
    \tag{$\lambda^*$ optimal for~\eqref{price-dual-min-restricted}}
    \gamma_{x}(x,\lambda)=\mu_x(\lambda)-\pi(x)
    &\overset{(*)}{\geq} \mu_x(\lambda^*)-\pi(x)\\
    &= \mu_x^{\LP}(\lambda^*)-\pi(x) \tag{$\mu_x(\cdot)=\mu^{\LP}_x(\cdot)$}\\
    &= \pi_x^{\LP}-\pi(x)= \iota_{x}(x). \tag{strong LP-duality}
  \end{align}
  Hence, we have shown the validity of \eqref{eq:general-formula-dual}.
  The last statement of the theorem now follows immediately
  by  \eqref{eq:general-formula-dual} and the fact that
  the inequality $(*)$ is tight for $\lambda$ precisely when $\mu_x(\lambda) =\mu_x(\lambda^*)$,
  i.e., if $\lambda$ is optimal for~\eqref{price-dual-min-restricted}.
\end{proof}

Theorem~\ref{thm:main-general} (Corollary~\ref{cor:main-general})
delivers a complete characterization of achievable $\Delta$-regret
\Walras equilibria by means of duality/integrality gaps of associated
parameterized optimization problems.
In the following, we apply this general theory to the important special
case of weakly monotone valuations.
\begin{definition}[Weakly Monotone Valuations and Upwards-Closed Strategy
  Spaces]
  A model $I$ has weakly monotone valuations and upwards-closed strategy
  spaces, if for all $i\in N$, we have
  \begin{align}
    \label{eq:monotone}
    &\forall x_i \in X_i,\forall y_i\in X_i:
    &x_i\leq y_i   \implies \pi_i(x_i)\leq \pi_i(y_i),
    \\
    \label{eq:strategy}
    &\forall x_i \in X_i, \forall y_i \in \Z^m:
    &   x_i\leq y_i\leq u \implies y_i \in X_i.
  \end{align}
\end{definition}
Condition~\eqref{eq:monotone} states that valuations are \emph{weakly
  monotone}, meaning that more items cannot decrease the utility.
Condition~\eqref{eq:strategy} allows us to assign more items (up to
capacity) to players $i\in N$.
This is equivalent to the assumption that undesirable items can be
disposed at zero cost.
Both conditions are standard in the literature on combinatorial \Walras
equilibria; see, e.g., \textcite{Blumrosen_Nisan_2007}.

When it comes to solving~\eqref{regret-opt}, it is important to
observe that we can without loss of generality
restrict the space of feasible strategy profiles to the set
\begin{equation}
  \label{eq:strat-up}
  \bar X(u) \define \Defset{x\in X}{\ell(x)=u},
\end{equation}
because any capacity-feasible $x$ with $\lambda\in\Lambda(x,u)$ can be transformed
to some $(x',\lambda)\in \bar X(u) \times \R^m_+$ by giving
unallocated items in $x$ to some players.
This does not increase the regret as valuations are monotone and any
item $j$ not fully allocated in $(x,\lambda)$, i.e., $\ell_j(x)< u_j$,
has price $\lambda_j=0$ by the fulfillment of
\eqref{eq:complementarity}.
Effectively, this assumption removes the complexity of the
complementary condition~\eqref{eq:complementarity}, i.e.,
$\lambda^\top (\ell(x)-u) = 0$, which is required for $\Delta$-regret
\Walras equilibria.
In particular, for any $x \in \bar X(u)$, the corresponding
$x$-parameterized
welfare maximization problem becomes now the classical
welfare maximization problem~\eqref{price-opt}.
We obtain as an immediate consequence of Theorem~\ref{thm:main-general}:
\begin{corollary}\label{cor:monotone}
  Let $I$ be a resource allocation model with weakly monotone valuations
  and upwards-closed strategy spaces.
  Then, the following statements hold.
  \begin{enumerate}
  \item For any $(x,\lambda)\in  \bar X(u)\times  \Lambda(x, u)$:
    $\Reg(x,\lambda)=\gamma(x,\lambda)\geq \iota(x)$.
  \item For any $x \in  \bar X(u)$, we have
    $\min_{\lambda\in\Lambda(x, u)}\Reg(x,\lambda)=\iota(x)$.
  \end{enumerate}
  Moreover, for $(x^*,\lambda^*)\in X(u)\times\Lambda(x^*,u)$, the following statements are equivalent.
  \begin{enumerate}[label=(\alph*),start=3]
  \item\label{enum:cor:opt1} $(x^*, \lambda^*)$ is primal-dual optimal
    for~\eqref{price-opt} and \eqref{price-dual-min}.
  \item\label{enum:cor:opt2} $(x^*, \lambda^*)$ is an optimal solution
    of~\eqref{regret-opt}.
  \end{enumerate}
\end{corollary}
Corollary~\ref{cor:monotone} includes the classical LP-characterization of the existence of
exact \Walras equilibria in~\textcite{Bikhchandani1997} as a special case.
\begin{corollary}[\cite{Bikhchandani1997}]
  For weakly monotone valuations
  and upwards-closed strategy spaces,  an exact ($0$-regret) \Walras equilibrium
  exists if and only if there exists $x\in \bar X(u)$ with $\iota(x)=0$.
\end{corollary}

  Under the assumptions of Corollary~\ref{cor:monotone}, the globally
  minimum regret equals $\iota^*$. For $i\in N$, let
    $V_i\define\max_{z_i\in X_i}\pi_i(z_i)$ denote the maximum value
    player~$i$ can attain, and let $V_{(1)}\geq\cdots\geq V_{(n)}$ denote
    these values sorted in nonincreasing order. A standard vertex-rounding argument,
  in the spirit of the Shapley--Folkman principle~\parencite{Starr69}, gives
  \begin{equation}\label{eq:fractional-player-bound}
    \iota^*\leq\sum_{h=1}^{\min\{m,n\}} V_{(h)}
    \leq\min\{m,n\}\max_{i\in N}V_i.
  \end{equation}
  Indeed, a basic optimal solution~$\alpha$ of the classical
    configuration LP~\eqref{LP:price-opt} (see Section~\ref{sec:weighted-seller}) has at most $n+m$ positive
  entries $\alpha_{i,l}$, since there are $n$ simplex equalities and $m$ coupling
  inequalities. Each player needs at least one positive entry, since
    $1^\top\alpha_i=1$, so at
  most $m$ players $i$ have a fractional
    vector~$\alpha_i$. Keep the integral
  players' bundles and assign the empty bundle to the fractional players.
  Nonnegative consumptions preserve feasibility, and the welfare loss is
  at most the sum of those players' $V_i$. This bounds $\iota^*$ even
  without monotonicity; its identification with globally minimum regret
  uses monotonicity and upward closure. For fixed $m$ and uniformly
  bounded individual values, the bound is independent of the number of
  players.

\section{Regret Bounds}
\label{sec:regret-bounds}

In what follows, we utilize the general theory derived so far to
give upper bounds on the minimal achievable regret.

  We first introduce the optimal-value function of the
  classical configuration~LP as a function of its right-hand
    side, i.e., of the capacity vector, while
  the strategy spaces and valuations remain fixed.
  It will be used both in Section~\ref{sec:weighted-seller}
  for a one-parameter family of objectives weighting
  seller regret against buyer regret and in
    Section~\ref{sec:complementarity-cost} for
  quantitative comparisons with the $x$-parameterized corresponding
  program. The classical configuration-LP relaxation is
\begin{align}
  \tag{LP}
  \label{LP:price-opt}
  \max_{\alpha} \quad
  &\sum_{i\in N}  \pi_{i}^\top  \alpha_i  \\
  \label{ineq:u-convex-finite} \text{s.t.} \quad
  &\ell(\alpha)\leq  u,\\
  &\alpha_i \in \Simp_i, \ i \in N. \nonumber
\end{align}

Consider an instance of~\eqref{LP:price-opt} with right-hand side
$u\in \Z_+^m$.
Let us denote by
\[ \rho:\R_+^m \to \R, \quad
  q \mapsto \max \Defset{\sum_{i\in N}\pi_i^\top\alpha_i}
  {\ell(\alpha) \leq q, \ \alpha_i \in \Simp_i, \ i\in N} \]
the \emph{optimal-value function} of~\eqref{LP:price-opt} with
respect to its right-hand side. In particular, $\rho(u)=\pi^{\LP}$
  is the optimal value of~\eqref{LP:price-opt}.\label{sym:rho}
  \subsection{Weighted Seller Regret and Finite Exact Penalties}
  \label{sec:weighted-seller}
  Fix $x\in X(u)$, and write $s(x)\define u-\ell(x)$ and
  $r(x)\define\min_{\lambda\in\Lambda(x,u)}\Reg(x,\lambda)$
  for its slack vector and minimum achievable regret, respectively.\label{sym:slack}
  For the seller-regret interpretation, we assume price-taking sellers
  with zero costs or reserve prices, one seller for each resource
    $j\in R$ who owns all $u_j$ units of~$j$. At the fixed price $\lambda_j\geq0$,
  let
  \[
    y_j^*(\lambda_j)\in
    \argmax_{y_j\in\{0,\ldots,u_j\}}\lambda_jy_j
    =\begin{cases}
        \{0,\ldots,u_j\}, &\text{if } \lambda_j=0,\\
        \{u_j\}, &\text{if } \lambda_j>0,
      \end{cases}
  \]
  denote a quantity that the seller of resource~$j$ wants to sell, whereas
  $\ell_j(x)$ is its realized sales quantity. Its regret is the revenue
  difference $\lambda_j y_j^*(\lambda_j)-\lambda_j\ell_j(x)$, and the
  aggregate seller regret is therefore
  \[
    S(x,\lambda)\define
    \sum_{j\in R}\lambda_j\bigl(y_j^*(\lambda_j)-\ell_j(x)\bigr)
    =\lambda^\top s(x).
  \]\label{sym:seller}
  The equality holds because selling all $u_j$ units maximizes revenue
  at every nonnegative price; when $\lambda_j=0$, all quantities give the
  same revenue and hence the same regret. The desired quantities are
  chosen independently at the given prices and may differ from realized
  sales. Each seller's regret is nonnegative, so $S(x,\lambda)=0$ is
  equivalent to $\lambda\in\Lambda(x,u)$.
  For $\tau\geq0$, consider the minimum weighted sum of buyer regret and
    seller regret over all nonnegative prices, where seller regret receives the
    weight~$\tau$:
  \begin{equation}\label{eq:weighted-regret}
    F_\tau(x)\define\min_{\lambda\geq0}
    \{\Reg(x,\lambda)+\tau S(x,\lambda)\}.
  \end{equation}

  \begin{theorem}[Weighted-seller characterization]\label{thm:weighted-regret}
    For every $x\in X(u)$ and $\tau\geq0$, define
    $q_\tau(x)\define\ell(x)+\tau s(x)$. Then the minimum in~\eqref{eq:weighted-regret}  is attained and the equality
    \begin{equation}\label{eq:weighted-value}
      F_\tau(x)=\rho(q_\tau(x))-\pi(x)
    \end{equation}
    holds.
    In particular, we have
    $F_1(x)=\iota(x)$.
  \end{theorem}
  The capacity vector $q_\tau(x)$ consists of the realized loads
    $\ell(x)$ plus the slack $s(x)=u-\ell(x)$ weighted with~$\tau$. Hence,
    it interpolates between $q_0(x)=\ell(x)$ and $q_1(x)=u$ and exceeds $u$
    on slack resources for $\tau>1$.
  \begin{proof}
    Set $v_i(\lambda)\define\max_{z_i\in X_i}
    \{\pi_i(z_i)-\lambda^\top z_i\}$.\label{sym:vi} The objective in
    \eqref{eq:weighted-regret} is
    \begin{align*}
             \Reg(x,\lambda)+\tau S(x,\lambda)
             &=\sum_{i\in N}v_i(\lambda)-\bigl(\pi(x)-\lambda^\top\ell(x)\bigr)
               +\tau\lambda^\top s(x)\\
             &=\sum_{i\in N}v_i(\lambda)
               +\lambda^\top q_\tau(x)-\pi(x).
           \end{align*}
         The first two terms form the Lagrangian dual function of the
         configuration LP with right-hand side $q_\tau(x)$, in which the coupling
           constraints $\ell(\alpha)\leq q_\tau(x)$ are dualized: for every
           $\lambda\geq0$,
           \[
             \max_{\alpha\in\prod_{i}\Simp_i}\Bigl\{\sum_{i\in N}\pi_i^\top\alpha_i
             -\lambda^\top\bigl(\ell(\alpha)-q_\tau(x)\bigr)\Bigr\}
             =\sum_{i\in N}v_i(\lambda)+\lambda^\top q_\tau(x),
           \]
           because a linear function attains its maximum over the simplex $\Simp_i$
           at a vertex, i.e., at a single bundle $z_i\in X_i$. Its primal is feasible
         because the deterministic configuration corresponding to~$x$ has load
         $\ell(x)\leq q_\tau(x)$, and its objective is bounded because each
         $\Simp_i$ is a finite-dimensional simplex.
         By strong LP duality, the Lagrangian dual
           $\min_{\lambda\geq0}\{\sum_{i\in N}v_i(\lambda)+\lambda^\top q_\tau(x)\}$
           attains its minimum at an optimal dual solution~$\lambda^*$ of the
           configuration LP with right-hand side $q_\tau(x)$, and its value is
           $\rho(q_\tau(x))$. Subtracting the constant $\pi(x)$ shows that $\lambda^*$
           attains the minimum in~\eqref{eq:weighted-regret} and
           proves~\eqref{eq:weighted-value}.
         At $\tau=1$, $q_1(x)=u$, giving $F_1(x)=\rho(u)-\pi(x)=\iota(x)$.
       \end{proof}

       We now define the maximum aggregate consumption of resource~$j$ by
       \begin{equation}\label{eq:max-aggregate}
         M_j\define\sum_{i\in N}\max_{z_i\in X_i}z_{ij}.
       \end{equation}
       Thus, $\ell_j(z)\leq M_j\leq n u_j$ for every $z\in X$, and the same
       bound holds for every fractional load, i.e.,
         $\ell_j(\alpha)=\sum_{i\in N}\sum_{l}\alpha_{i,l}z_{ij}^l\leq M_j$ for every
         $\alpha\in\prod_{i\in N}\Simp_i$. Set
       \begin{equation}\label{eq:penalty-threshold}
         \tau_0(x)\define
         \max_{j\in\Rleq(x)}
         \frac{M_j-\ell_j(x)}{u_j-\ell_j(x)},
       \end{equation}
       where the maximum over the empty set is defined as zero.
       For each resource~$j$ with slack at~$x$, the ratio compares the
         largest possible additional consumption $M_j-\ell_j(x)$ beyond the
         realized load with the actual slack $u_j-\ell_j(x)$. Thus, $\tau_0(x)$ is
         the smallest weight such that the $\tau$-weighted slack covers every
         possible additional consumption, i.e., $q_{\tau,j}(x)\geq M_j$ for all
         $j\in\Rleq(x)$ and $\tau\geq\tau_0(x)$.

       \begin{theorem}[Finite exact penalty]\label{thm:finite-penalty}
         For every $x\in X(u)$ and $\tau\geq\tau_0(x)$,
         \begin{equation}\label{eq:finite-penalty-value}
           F_\tau(x)=r(x).
         \end{equation}
         If $\tau>\tau_0(x)$, every minimizer of~\eqref{eq:weighted-regret}
         belongs to $\Lambda(x,u)$. Since consumptions and capacities are integral (see Remark~\ref{rem:integrality}),
         \begin{equation}\label{eq:uniform-penalty-threshold}
           \tau_0(x)\leq1+(n-1)u_{\max},
             \quad\text{where } u_{\max}\define\max_{j\in R}u_j.
         \end{equation}
       \end{theorem}
       \begin{proof}
         For $j\in\Req(x)$, $q_{\tau,j}(x)=u_j$${}=\ell_j(x)$. For $j\in\Rleq(x)$ and
         $\tau\geq\tau_0(x)$, the definition of $\tau_0(x)$ gives
           \[
             q_{\tau,j}(x)=\ell_j(x)+\tau s_j(x)\geq\ell_j(x)+\tau_0(x)s_j(x)
             \geq\ell_j(x)+\bigl(M_j-\ell_j(x)\bigr)=M_j,
           \]
           so  its configuration-LP
         capacity constraint is redundant, since $\ell_j(\alpha)\leq M_j$ for every
           $\alpha\in\prod_{i\in N}\Simp_i$. The remaining constraints $\ell_j(\alpha)\leq\ell_j(x)$, $j\in\Req(x)$, define the
         $x$-parameterized configuration LP~\eqref{LP:price-opt-tilde-restricted}, so
           $\rho(q_\tau(x))=\pi^{\LP}_x$.
         Hence, Theorem~\ref{thm:weighted-regret} and
           Theorem~\ref{thm:main-general} yield
           \[
             F_\tau(x)\overset{\text{Thm.~\ref{thm:weighted-regret}}}{=}\rho(q_\tau(x))-\pi(x)
             =\pi^{\LP}_x-\pi(x)=\iota_x(x)
             \overset{\text{Thm.~\ref{thm:main-general}}}{=}
             \min_{\lambda\in\Lambda(x,u)}\Reg(x,\lambda)=r(x),
           \]
           which proves~\eqref{eq:finite-penalty-value}.

         For the assertion about every minimizer, take $\lambda\geq0$ and define
         $\bar\lambda_j=\lambda_j$ for $j\in\Req(x)$ and $\bar\lambda_j=0$
         otherwise. For every player~$i$, let $\bar z_i\in X_i$ attain $v_i(\bar\lambda)$. Since
           $\lambda_j-\bar\lambda_j=\lambda_j$ for $j\in\Rleq(x)$ and zero otherwise,
         \begin{align*}
                  v_i(\bar\lambda)-v_i(\lambda)
                  &\leq\bigl(\pi_i(\bar z_i)-\bar\lambda^\top\bar z_i\bigr)
                    -\bigl(\pi_i(\bar z_i)-\lambda^\top\bar z_i\bigr)
                    =\sum_{j\in\Rleq(x)}\lambda_j\bar z_{ij}\\
                  &\leq\max_{z_i\in X_i}\sum_{j\in\Rleq(x)}\lambda_j z_{ij}
                    \leq\sum_{j\in\Rleq(x)}\lambda_j
                    \max_{z_i\in X_i}z_{ij}.
                \end{align*}
              Summing over $i\in N$, and using $\Reg(x,\lambda)=\sum_{i\in N}v_i(\lambda)-\pi(x)+\lambda^\top\ell(x)$ as well as
                $(\bar\lambda-\lambda)^\top\ell(x)=-\sum_{j\in\Rleq(x)}\lambda_j\ell_j(x)$, we obtain the first line below. For the second line, we negate the first one,
                add $\tau S(x,\lambda)=\tau\sum_{j\in\Rleq(x)}\lambda_js_j(x)$ on both sides, and use $\tau S(x,\bar\lambda)=0$.
              \begin{align*}
                \Reg(x,\bar\lambda)-\Reg(x,\lambda)
                &\leq\sum_{j\in\Rleq(x)}\lambda_j(M_j-\ell_j(x)),\\
                \Reg(x,\lambda)+\tau S(x,\lambda)-\Reg(x,\bar\lambda)-\tau S(x,\bar\lambda)
                &\geq\sum_{j\in\Rleq(x)}\lambda_j
                  \bigl(\tau s_j(x)-(M_j-\ell_j(x))\bigr).
              \end{align*}
              For $\tau>\tau_0(x)$, every coefficient $\tau s_j(x)-(M_j-\ell_j(x))$, $j\in\Rleq(x)$, in the last sum is positive by the definition of $\tau_0(x)$ and $s_j(x)>0$.
              Hence, if $\lambda_j>0$ for some $j\in\Rleq(x)$, the weighted objective
                in~\eqref{eq:weighted-regret} is strictly smaller at $\bar\lambda$ than at
                $\lambda$, so $\lambda$ is not a minimizer. Thus every minimizer vanishes
                on $\Rleq(x)$, i.e., belongs to $\Lambda(x,u)$. Finally, if $j\in\Rleq(x)$,
              integrality gives $s_j(x)\geq1$, and
              \[
                \frac{M_j-\ell_j(x)}{s_j(x)}
                =1+\frac{M_j-u_j}{s_j(x)}
                \leq1+(n-1)u_j\leq1+(n-1)u_{\max}.
              \]
              This proves~\eqref{eq:uniform-penalty-threshold}.
            \end{proof}

            At $\tau=0$ only buyer regret is minimized, without any restriction on
              the prices of slack resources, i.e., complementarity is not enforced; at $\tau=1$ buyer and seller
            regret receive equal weight, recovering the classical duality-gap
            objective~\parencite{GHP07,HB17}. A finite sufficiently large weight
            recovers exact market clearing at the prescribed allocation.
            For $\tau=\tau_0(x)$, the optimal value equals $r(x)$ and every
              complementary price vector attaining $r(x)$ is a minimizer, but there
              may also be minimizers that are not complementary; see
              Example~\ref{ex:repair-sharp}. Only for $\tau>\tau_0(x)$,
              Theorem~\ref{thm:finite-penalty} guarantees that every minimizer is
              complementary.

            For fixed $x$, the function $\tau\mapsto F_\tau(x)$ is concave and
            nondecreasing: it is the pointwise minimum of affine functions of $\tau$
            with nonnegative slopes $S(x,\lambda)$. Theorem~\ref{thm:finite-penalty}
            shows that it is constant for $\tau\geq\tau_0(x)$. Thus the finite-penalty
            result specifies an exact penalty for the constraint $S(x,\lambda)=0$,
            with an explicit threshold depending only on consumptions and capacities.

            \subsection{The Cost of Enforcing Complementarity}\label{sec:complementarity-cost}
            The $x$-parameterized configuration LP~\eqref{LP:price-opt-tilde-restricted} can also be written using a finite capacity
            vector, namely $q_{\tau_0(x)}(x)$: as shown in the proof of
              Theorem~\ref{thm:finite-penalty}, the constraints of slack resources are
              redundant for this right-hand side, so $\rho(q_{\tau_0(x)}(x))=\pi_x^{\LP}$.
            Theorem~\ref{thm:main-general}
            therefore gives the decomposition
            \begin{equation}\label{eq:regret-decomposition}
              r(x)=\underbrace{\rho(q_{\tau_0(x)}(x))-\rho(u)}_{\text{cost of enforcing complementarity}}
              \hspace{-0.6cm} + \quad \underbrace{\rho(u)-\pi(x)}_{\iota(x)}.
            \end{equation}
            The first term is nonnegative, since the feasible set of~\eqref{LP:price-opt} is contained in that of~\eqref{LP:price-opt-tilde-restricted}, and equals
            $F_{\tau_0(x)}(x)-F_1(x)$ by Theorems~\ref{thm:weighted-regret} and~\ref{thm:finite-penalty}. It measures the increase in the minimum
            opportunity-cost objective when zero seller regret is imposed, relative
            to equal weighting of buyer and seller regret. It is not a comparison
            with the minimum of buyer regret alone, which is $F_0(x)$.
            If all capacities bind at $x$, then $q_{\tau_0(x)}(x)=u$ and the first term
            vanishes. The decomposition is a direct consequence of the
            characterization in Theorem~\ref{thm:main-general}; its role is to distinguish the allocation-relative
            classical LP gap from the additional effect of complementarity.

            We illustrate the decomposition~\eqref{eq:regret-decomposition} on
              the triangle instance of Example~\ref{ex:nonexistence}. The example shows
              that the cost of enforcing complementarity can be as large as the classical
              LP gap~$\iota(x)$, even at a welfare-optimal allocation, and that the
              classical LP gap alone does not determine the minimum regret.
            \begin{rstexample}{\ref{ex:nonexistence}}[continued]
              Use the original three-player instance and fix
              $x=(\{1,2\},\varnothing,\varnothing)$, with $\pi(x)=1$.
              The classical configuration LP has value $3/2$: place weight $1/2$ on
              each player's desired pair and $1/2$ on its empty bundle. The prices
              $(1/2,1/2,1/2)$ and zero utility variables give a matching dual certificate.
              Only the third capacity constraint is removed in the $x$-parameterized LP.
              Its value is~2: assign the empty bundle to player~1 and the desired pairs
              to players~2 and~3, each with weight~1. This is optimal since every valued
              bundle consumes resource~1 or~2, whose total capacity is~2.
              Thus,
              \[
                \rho(q_{\tau_0(x)}(x))-\rho(u)=2-3/2=1/2,
                \quad \iota(x)=3/2-1=1/2.
              \]
              Equation~\eqref{eq:regret-decomposition} gives $r(x)=1$, attained at prices
              $(1/2,1/2,0)$. This is the minimum regret for the fixed allocation~$x$;
              the globally minimum regret of the instance remains~$1/2$.
              It is attained by the allocation $x'=(\{1,2\},\{3\},\varnothing)$ of
                Example~\ref{ex:nonexistence}, which has the same welfare and the same
                classical LP gap $\iota(x')=1/2$ as~$x$. Since $x'$ sells all items, its
                cost of enforcing complementarity vanishes, and $r(x')=1/2<1=r(x)$.
                For~$x$, the classical dual prices $(1/2,1/2,1/2)$ yield buyer regret~$0$
                and seller regret~$1/2$ and thus attain $F_1(x)=\iota(x)=1/2$, but they
                violate complementarity, since the unsold item~3 has a positive price.
            \end{rstexample}

            \subsection{Regret Bounds via Duality Gaps}
            \label{sec:price-repair}
            We now bound the effect of setting all prices on slack resources to zero.
            The bound uses the actual slack and maximum aggregate consumption, and
            recovers the uniform estimate in terms of $n$ and $u_{\max}$.

            The following lemma shows that the classical (unparameterized)
              duality gap $\gamma(x,\lambda)$ of a feasible allocation~$x$ and
              arbitrary nonnegative prices~$\lambda$ splits into the total buyer regret
              and the seller regret at $(x,\lambda)$.
            \begin{lemma}\label{lem:bound-slack}
              For $x\in X(u)$ and $\lambda\geq0$,
              \begin{equation}\label{eq:slack}
                \gamma(x,\lambda)=\Reg(x,\lambda)+S(x,\lambda).
              \end{equation}
              In particular $S(x,\lambda)\leq\gamma(x,\lambda)$ and, since slacks are integral, $\sum_{j\in\Rleq(x)}\lambda_j\leq\gamma(x,\lambda)$.
            \end{lemma}
            \begin{proof}
              The identity follows by expanding $\mu(\lambda)-\pi(x)$ and using
              $u=\ell(x)+s(x)$:
                \begin{align*}
                  \gamma(x,\lambda)=\mu(\lambda)-\pi(x)
                  &=\max_{z\in X}\bigl\{\pi(z)-\lambda^\top\ell(z)\bigr\}
                    -\bigl(\pi(x)-\lambda^\top\ell(x)\bigr)+\lambda^\top\bigl(u-\ell(x)\bigr)\\
                  &=\Reg(x,\lambda)+\lambda^\top s(x)=\Reg(x,\lambda)+S(x,\lambda).
                \end{align*} Regret is nonnegative, proving the first inequality.
              For the second, use $s_j(x)\geq1$ on each slack resource.
            \end{proof}

            In particular, weak duality gives the welfare-loss bound
            \[
              \mathrm{OPT}-\pi(x)\leq\Reg(x,\lambda)+S(x,\lambda).
            \]
            This is the fixed-supply counterpart of the bound-form first welfare
            theorem by \textcite[Theorem~1]{MilgromWatt26}: buyer regret is a
            rationing loss, and the priced excess of fixed supply over consumption
            is $S(x,\lambda)$. Our zero-cost seller interpretation identifies this
            same quantity with foregone sales revenue. With complementary prices,
            welfare loss is at most buyer regret, although a regret-minimizing
            allocation need not maximize welfare, as Proposition~\ref{prop:welfare-regret}
            shows.

            Lemma~\ref{lem:bound-slack} shows that the classical duality gap
              controls buyer and seller regret simultaneously. We now use this to bound
              the regret that arises if we start from arbitrary prices~$\lambda\geq0$
              and enforce complementarity by setting the prices of all slack resources
              to zero. The resulting loss depends on how strongly the slack resources
              could be over-demanded relative to their slack, which is measured by the
              following factor.
            Set $(a)_+\define\max\{a,0\}$ and
            \begin{equation}\label{eq:slack-factor}
              \eta(x)\define\max_{j\in\Rleq(x)}
              \frac{(M_j-u_j)_+}{u_j-\ell_j(x)},
            \end{equation}
            with maximum zero for an empty index set.
            For every resource~$j$ with slack at~$x$, the ratio compares the
              largest possible over-consumption $(M_j-u_j)_+$ beyond the capacity with
              the slack $u_j-\ell_j(x)$ at~$x$. Since
              $(M_j-\ell_j(x))/s_j(x)=1+(M_j-u_j)/s_j(x)$, we have
              $\tau_0(x)\leq1+\eta(x)$, with equality if $\Rleq(x)\neq\emptyset$ and
              $M_j\geq u_j$ for all $j\in\Rleq(x)$.
            \begin{theorem}[Slack-dependent price adjustment]\label{thm:suf-condition}
              For any $x\in X(u)$ and $\lambda\geq0$, define
              \[
                \bar\lambda_j=\begin{cases}
                  \lambda_j,&\ell_j(x)=u_j,\\
                  0,&\ell_j(x)<u_j.
                \end{cases}
              \]
              Then $\bar\lambda\in\Lambda(x,u)$ and
              \begin{align}
                \Reg(x,\bar\lambda)
                &\leq\gamma(x,\lambda)
                  +\sum_{j\in\Rleq(x)}\lambda_j(M_j-u_j)\label{eq:repair-intermediate}\\*
                &\leq(1+\eta(x))\gamma(x,\lambda)\label{eq:repair-refined}\\*
                &\leq\bigl(1+(n-1)u_{\max}\bigr)\gamma(x,\lambda).
                  \label{eq:repair-uniform}
              \end{align}
            \end{theorem}
            \begin{proof}
              The price comparison in the proof of Theorem~\ref{thm:finite-penalty}
              and Lemma~\ref{lem:bound-slack} give
              \begin{align*}
                \Reg(x,\bar\lambda)
                &\leq\Reg(x,\lambda)
                  +\sum_{j\in\Rleq(x)}\lambda_j(M_j-\ell_j(x))\\
                &=\gamma(x,\lambda)
                  +\sum_{j\in\Rleq(x)}\lambda_j(M_j-u_j),
              \end{align*}
              proving~\eqref{eq:repair-intermediate}. The last sum is at most
              $\eta(x)S(x,\lambda)\leq\eta(x)\gamma(x,\lambda)$, which proves
              \eqref{eq:repair-refined}. Finally $M_j\leq n u_j$ and integer slack
              give $\eta(x)\leq(n-1)u_{\max}$, proving~\eqref{eq:repair-uniform}.
            \end{proof}

            In particular, taking a classical primal-dual optimal pair gives the
            global existence bound 
              \begin{equation}\label{eq:global-existence-bound}
                \min_{x\in X(u)}r(x)\leq\bigl(1+(n-1)u_{\max}\bigr)\iota^*.
              \end{equation}
            The refined factor is invariant under changing the units of any resource:
            scaling its consumptions and capacity by the same positive constant
            leaves the corresponding ratio in~\eqref{eq:slack-factor} unchanged.
            When the maxima $M_j$ are unavailable, the computable upper bound
            \[
              \eta(x)\leq(n-1)\max_{j\in\Rleq(x)}
              \frac{u_j}{u_j-\ell_j(x)}
            \]
            uses only the allocation and capacities, again with maximum zero when
            the index set is empty.

            \begin{example}[Sharpness for a prescribed allocation]\label{ex:repair-sharp}
              Let $n\geq2$, let the single resource have integer capacity $U\geq2$,
              and set
              \[
                X_1=\{0,U-1,U\},\quad X_i=\{0,U\}\ (i=2,\ldots,n),
                \quad \pi_i(q)=q/U\ (i\in N).
              \]
              Fix $x=(U-1,0,\ldots,0)$ and $\lambda=1/U$. Every strategy of every player has utility zero at~$\lambda$, so $\Reg(x,\lambda)=0$ and
              $\gamma(x,\lambda)=S(x,\lambda)=1/U$.
              Since $x$ leaves one unit unsold, complementarity forces the only
              admissible price to be zero. At that price, player~1 has regret $1/U$
              and every other player has regret~1. Therefore,
              \[
                r(x)=n-1+1/U
                =\bigl(1+(n-1)U\bigr)\gamma(x,\lambda).
              \]
              Here $M=nU$ and $\eta(x)=(n-1)U$, so both
              \eqref{eq:repair-refined} and~\eqref{eq:repair-uniform} are tight.
              No smaller uniform factor is possible for guarantees preserving every
              prescribed allocation. This example does not establish sharpness of the
              global existence bound~\eqref{eq:global-existence-bound}: assigning all $U$ units to one player and setting
              price $1/U$ gives an exact equilibrium.
              Finally, $\tau_0(x)=1+(n-1)U$. At this threshold every price in
              $[0,1/U]$ minimizes the weighted objective in~\eqref{eq:weighted-regret} with $\tau=\tau_0(x)$, so the strict inequality in
              the assertion about every minimizer in Theorem~\ref{thm:finite-penalty}
              cannot in general be replaced by a weak inequality.
            \end{example}

            \begin{remark}[Integral and real consumptions]\label{rem:integrality}
                In our model, consumptions and capacities are integral; see
                Section~\ref{sec:model}. For the bounds in Section~\ref{sec:regret-bounds},
                integrality is used only via $s_j(x)\geq1$ for $j\in\Rleq(x)$, namely
                in~\eqref{eq:uniform-penalty-threshold}, in the last statement of
                Lemma~\ref{lem:bound-slack}, and in~\eqref{eq:repair-uniform} and~\eqref{eq:global-existence-bound}; the
                estimates~\eqref{eq:uniform-penalty-threshold}, \eqref{eq:repair-uniform}, and~\eqref{eq:global-existence-bound}
                additionally use the individual capacity bounds $z_{ij}\leq u_j$ for all
                $z_i\in X_i$, which give $M_j\leq nu_j$. Theorems~\ref{thm:main-general}, \ref{thm:weighted-regret},
                and~\ref{thm:finite-penalty} (except~\eqref{eq:uniform-penalty-threshold})
                as well as~\eqref{eq:repair-intermediate} and~\eqref{eq:repair-refined}
                remain valid for finite strategy spaces
                $X_i\subseteq\R_+^m$ and capacities $u\in\R_+^m$.
                For the seller-regret interpretation with real capacities, replace
                the quantity domain $\{0,\ldots,u_j\}$ by $[0,u_j]$.
                The full-utilization arguments in Section~\ref{sec:GenVal}, including
                Corollary~\ref{cor:monotone}, also use integrality to assign all remaining
                capacity and are not covered by this extension.
              \end{remark}


%% file: poly-time-approx-schemes.tex
\section{Polynomial-Time Approximation Algorithms}\label{sec:algorithms}
Corollary~\ref{cor:monotone} and Theorem~\ref{thm:suf-condition} give a
black-box reduction from welfare algorithms with additive LP guarantees
to algorithms for computing low-regret \Walras equilibria.

We first specify the computational assumptions. An instance
$I=(N,R,u,X,\pi)$ has a finite encoding of length $L=\langle I\rangle$\label{sym:encoding}
containing $n$, $m$, the binary-encoded capacities, and a succinct
description of the strategy spaces and valuations. All valuation values
are rational, with a known encoding-length bound polynomial in~$L$.
We assume that we have a valuation oracle and a demand oracle for each
player. Given a bundle $x_i\in X_i$, the former returns its value
$\pi_i(x_i)$. Given rational prices $\lambda\in\R_+^m$, the latter returns 
\[
 x_i(\lambda)\in\argmax_{x_i\in X_i}
       \{\pi_i(x_i)-\lambda^\top x_i\}.
\]
Both oracles run in time polynomial in $L$ and the encoding length of
their queries. These rational-data assumptions apply to the computational
results here and in Section~\ref{sec:complexity}; the structural results
allow real valuations.

The dual of the classical configuration LP is
\begin{align}
 \tag{D-LP}\label{LP-dual}
 \min_{\nu,\lambda}\quad &\sum_{i\in N}\nu_i+\sum_{j\in R}u_j\lambda_j\\
 \text{s.t.}\quad &\nu_i+\sum_{j\in R}x_{i,j}^l\lambda_j
             \geq\pi_i(x_i^l), &&i\in N,\ l=1,\ldots,k_i,\notag\\
 &\nu_i\in\R, &&i\in N,\notag\\
 &\lambda_j\geq0, &&j\in R.\notag
\end{align}
It has $n+m$ variables and possibly exponentially many constraints.
First checking $\lambda\geq0$ and then evaluating each player's demand
provides a separation oracle: a point $(\nu,\lambda)$ with $\lambda\geq0$ violates some constraint
if and only if $\nu_i<\pi_i(x_i(\lambda))-\lambda^\top x_i(\lambda)$ for
some $i\in N$, and then the constraint of the demanded bundle
$x_i(\lambda)$ is violated. By the equivalence of separation and
optimization, the ellipsoid method with this separation oracle computes an
optimal solution of~\eqref{LP-dual} and, hence, the LP value $\rho(u)$ in
polynomial time; the coefficient-length bounds above guarantee the required
bound on the encoding length of the constraints; see~\textcite{GroetschelLovaszSchrijver1993}.

The second ingredient is an additive adaptation of the multiplicative
integrality-gap-verification concept of
\textcite[Definition~2.1]{ElbassioniFS13}.
\begin{definition}[Additive LP-gap verification]\label{def:additive-verification}
An algorithm $\ALG$ verifies an additive configuration-LP gap of at most
$\alpha(I)\geq0$ on a class $\mathcal I$ if, for every $I\in\mathcal I$,
it returns an allocation, denoted by $\ALG(I)\define$ $\tilde x\in X(u)$, satisfying
\[
 \iota(\tilde x)=\rho(u)-\pi(\tilde x)\leq\alpha(I).
\]
\end{definition}
Here $\alpha(I)$ may depend on the instance. This is an additive
guarantee relative to the LP optimum, not a multiplicative approximation
ratio relative to the integer optimum.

\begin{theorem}\label{thm:PolyAlg}
Let $\mathcal I$ satisfy the encoding and oracle assumptions above, and
let $\ALG$ be a polynomial-time algorithm verifying an additive
configuration-LP gap of at most $\alpha(I)$. Then:
\begin{enumerate}
\item\label{bb-alg-item-1} For monotone valuations and upwards-closed
strategy spaces, Algorithm~\ref{alg3} runs in polynomial time and computes a $\Delta$-regret
\Walras equilibrium with $\Delta\leq\alpha(I)$.
\item\label{bb-alg-item-2} For general valuations and strategy spaces,
Algorithm~\ref{alg4} runs in polynomial time and computes a $\Delta$-regret \Walras equilibrium
with $\Delta\leq\alpha(I)(1+(n-1)u_{\max})$, where $u_{\max}=\max_{j\in R}u_j$.
\end{enumerate}
\end{theorem}
\begin{proof}
Let $\bar x=\ALG(I)$ and let $\lambda^*$ be an optimal price vector of
\eqref{LP-dual}. Strong duality gives
\[
 \gamma(\bar x,\lambda^*)=\rho(u)-\pi(\bar x)\leq\alpha(I).
\]
For part~\ref{bb-alg-item-1}, give all unallocated units to any one player.
The resulting $\tilde x\geq\bar x$ belongs to $\bar X(u)$ by upward
closure, and $\pi(\tilde x)\geq\pi(\bar x)$ by monotonicity.
Corollary~\ref{cor:monotone} therefore yields
$\Reg(\tilde x,\lambda^*)=\rho(u)-\pi(\tilde x)\leq\alpha(I)$.
For part~\ref{bb-alg-item-2}, keep $\bar x$ and set to zero precisely
the prices of resources with slack at $\bar x$. The stated guarantee
follows from Theorem~\ref{thm:suf-condition}. These are exactly the steps of
Algorithms~\ref{alg3} and~\ref{alg4}, respectively. Both call $\ALG$ once,
compute $\lambda^*$ in polynomial time as described after~\eqref{LP-dual},
and perform a modification that clearly takes polynomial time.
\end{proof}

\begin{algorithm}
\caption{Low-regret equilibrium for monotone valuations and upwards-closed strategy spaces.}
\label{alg3}
\begin{algorithmic}[1]
\State Compute $\bar x=\ALG(I)$.
\State Give all unused units to one player, obtaining $\tilde x\in\bar X(u)$.
\State Compute optimal prices $\lambda^*$ for~\eqref{LP-dual}.
\State \Return $(\tilde x,\lambda^*)$, with regret at most $\alpha(I)$.
\end{algorithmic}
\end{algorithm}

\begin{algorithm}
\caption{Low-regret equilibrium for general valuations and strategy spaces.}
\label{alg4}
\begin{algorithmic}[1]
\State Compute $\tilde x=\ALG(I)$.
\State Compute optimal prices $\lambda^*$ for~\eqref{LP-dual}.
\State Set $\bar\lambda_j=\lambda_j^*$ if $\ell_j(\tilde x)=u_j$, and $\bar\lambda_j=0$ otherwise.
\State \Return $(\tilde x,\bar\lambda)$, with regret at most $\alpha(I)(1+(n-1)u_{\max})$.
\end{algorithmic}
\end{algorithm}

\begin{remark}\label{rem:algorithm-slack}
For the general-valuation algorithm, Theorem~\ref{thm:suf-condition}
also gives the allocation-dependent guarantee
\[
 \Reg(\tilde x,\bar\lambda)\leq\alpha(I)\left(
 1+(n-1)\max_{j\in\Rleq(\tilde x)}
              \frac{u_j}{u_j-\ell_j(\tilde x)}\right),
\]
where the maximum is zero if every capacity is binding. This factor can
be evaluated directly from the returned allocation. If the coordinate
maxima in~\eqref{eq:max-aggregate} are also available, the factor
$1+\eta(\tilde x)$ can be used instead.
\end{remark}

\begin{corollary}[Conversion of a multiplicative LP guarantee]\label{cor:ratio-conversion}
Under the  computational assumptions of Theorem~\ref{thm:PolyAlg}, suppose that a polynomial-time
algorithm returns $\tilde x$ with $\pi(\tilde x)\geq\rho(u)/\beta$,
where $\beta\geq1$. Then Theorem~\ref{thm:PolyAlg} applies with
$\alpha(I)=(1-1/\beta)\rho(u)$.
\end{corollary}
\begin{proof}
Subtract the assumed welfare guarantee from $\rho(u)$. The empty
allocation ensures $\rho(u)\geq0$, so the resulting additive bound is
nonnegative.\qedhere
\end{proof}
The factor must be verified against the configuration LP. A guarantee
$\pi(\tilde x)\geq\mathrm{OPT}(I)/\beta$ alone does not imply the premise
of Corollary~\ref{cor:ratio-conversion}.

For randomized welfare algorithms, the monotone-case argument applies
outcome by outcome after allocating all unused items. Consequently,
if $\mathbb E[\pi(\bar x)]\geq c\rho(u)$, the returned feasible allocation
and complementary prices satisfy
\[
 \mathbb E[\Reg(\tilde x,\lambda^*)]\leq(1-c)\rho(u).
\]
This is an expected-regret guarantee; it does not assert the same bound
for every realization.
For single-unit markets with the full bundle domain, the randomized
configuration-LP rounding of \textcite{Feige09} gives $c=1-1/e$ for
monotone XOS valuations (and hence for monotone submodular valuations).
It therefore yields expected regret at most $\rho(u)/e$.
For monotone subadditive valuations, the same reference gives
integrality gap at most~2, so an equilibrium with regret at most
$\rho(u)/2$ exists. Its polynomial-time rounding with an arbitrarily
small loss gives, for any fixed $\varepsilon\in(0,1/2)$, expected regret
at most $(1/2+\varepsilon)\rho(u)$.
These algorithmic statements assume the demand and valuation oracles
specified above; a configuration-LP solution with polynomial support
can be recovered by the standard separation--optimization reduction.

%% file: lower-bounds.tex
\section{Lower Bounds}\label{sec:complexity}

Corollary~\ref{cor:nec} translates additive integrality gaps into lower
bounds on achievable regret. We first illustrate this for network
valuations. We then use hardness of approximating welfare values to
exclude uniformly small regret, following the approach to exact
equilibria by~\textcite{Roughgarden:2015}.

\subsection{Lower Bounds via Integrality Gaps}
The approximability and the integrality gap of social welfare problems
have been studied extensively; see, e.g., \textcite{FeldmanGL15}.
We give an explicit directed version of the grid construction by
\textcite[Section~6, Proposition~6.1]{GargVY97}, whose original graph is
undirected.

Let $G=(V,E)$ be a directed acyclic graph, where $V$ and $E$ are the node
and arc sets, and let $u\in\Z_+^E$ denote its capacities. Player $i$
has a source $s_i$ and a destination~$t_i$, with no arc entering $s_i$ and no arc leaving
$t_i$. Its finite strategy set is
\[
 X_i=\left\{x_i\in\Z_+^E:\ 0\leq x_i\leq u,\quad
 \sum_{e\in\delta^+(v)}x_{i,e}-\sum_{e\in\delta^-(v)}x_{i,e}=0
 \ \ (v\notin\{s_i,t_i\})\right\}.
\]
Here $\delta^+(v)$ and $\delta^-(v)$ denote the arcs leaving and entering
$v$. The valuation is the net flow
$\pi_i(x_i)=\val(x_i)=\sum_{e\in\delta^+(s_i)}x_{i,e}$, which equals
the flow reaching $t_i$. Thus $0\in X_i$, every individual strategy is
capacity bounded, and the welfare problem is
\begin{align}
 \tag{MIFP}\label{mif}
 \max_{x\in\prod_iX_i}\quad &\sum_{i\in N}\val(x_i)\\
 \text{s.t.}\quad &\sum_{i\in N}x_{i,e}\leq u_e,\qquad e\in E.\notag
\end{align}

\begin{proposition}\label{prop:directed-flow-gap}
For every $n\geq3$, there is a planar directed acyclic network with unit
capacities and $n$ players for which the integral welfare optimum is~1
and the configuration-LP value is at least $n/2$.
\end{proposition}
\begin{proof}
The $n$ designated source--destination paths can be drawn in a grid such
that every two of them share exactly one unit-capacity arc.
In fact, any two source--destination paths of different players
share an arc. Hence, only
one player can route flow integrally, whereas routing every designated
path with configuration weight~$1/2$ is feasible for the configuration
LP. The detailed construction and proof are given in
Appendix~\ref{app:directed-flow-gap}.
\end{proof}

\begin{corollary}\label{cor:directed-regret}
For every $n\geq3$, there are instances with network valuations such that
every feasible allocation and complementary price vector have total
regret at least $n/2-1$.
\end{corollary}
\begin{proof}
	This follows immediately by
Proposition~\ref{prop:directed-flow-gap} and Corollary~\ref{cor:nec}.\qedhere
\end{proof}
The construction also admits polynomial-time demand queries: each
nonzero strategy has value~1, so a query reduces to a shortest directed
source--destination path under the given arc prices, compared with the empty
flow. The lower bound is therefore compatible with efficient individual
optimization.

\subsection{Complexity-Theoretic Lower Bounds}
The preceding lower bound concerns explicit instances. The following
argument concerns a class of instances and hardness of approximating
the \emph{optimal welfare value}. Note that computing this value approximately
and producing an approximately optimal integral allocation are different
tasks; the proof uses only the former.

\begin{theorem}\label{thm:ConnectionNPComp}
Let $\mathcal I$ satisfy the rational encoding and polynomial-time oracle
assumptions of Section~\ref{sec:algorithms}, and assume $P\ne NP$.
\begin{enumerate}
\item\label{item:relative-hardness} Fix a rational $\varepsilon\in[0,1)$. If no polynomial-time
algorithm can return a number $v(I)$ satisfying
\[
 (1-\varepsilon)\mathrm{OPT}(I)\leq v(I)\leq\mathrm{OPT}(I)
 \qquad(I\in\mathcal I),
\]
then some $I\in\mathcal I$ admits no $\varepsilon\rho(u)$-regret
\Walras equilibrium.
\item\label{item:additive-hardness} Fix $\Delta\geq0$. If no polynomial-time algorithm can return a
number $v(I)$ satisfying $|v(I)-\mathrm{OPT}(I)|\leq\Delta$ for all
$I\in\mathcal I$, then some instance in $\mathcal I$ admits no
$\Delta$-regret \Walras equilibrium.
\end{enumerate}
\end{theorem}
\begin{proof}
The separation oracle for~\eqref{LP-dual} computes $\rho(u)$ in polynomial
time. If every instance admitted an $\varepsilon\rho(u)$-regret equilibrium,
Corollary~\ref{cor:nec} would give
\[
 0\leq\rho(u)-\mathrm{OPT}(I)\leq\varepsilon\rho(u).
\]
Thus $v(I)=(1-\varepsilon)\rho(u)$ would satisfy
$(1-\varepsilon)\mathrm{OPT}(I)\leq v(I)\leq\mathrm{OPT}(I)$,
contradicting the premise of part~\ref{item:relative-hardness}. For part~\ref{item:additive-hardness}, the same argument gives
$0\leq\rho(u)-\mathrm{OPT}(I)\leq\Delta$, and returning $\rho(u)$
contradicts the assumed additive-error hardness.\qedhere
\end{proof}

We briefly explain why we regard part~\ref{item:relative-hardness}
as the more meaningful statement. Multiplying all valuations of an instance
by a constant $K>0$ multiplies the optimal welfare, the LP value $\rho(u)$,
the integrality gap, and the minimum regret by~$K$ (for the latter, scale
the prices by~$K$ as well). The relative threshold $\varepsilon\rho(u)$ in
part~\ref{item:relative-hardness} scales in the same way, so the statement
does not depend on the choice of units. A fixed additive threshold~$\Delta$
in part~\ref{item:additive-hardness}, in contrast, does not scale. To see
the consequences, consider a class with integral valuations that contains,
with every instance, also all instances obtained by multiplying the
valuations by a positive integer~$K$. On such a class, the following two
observations hold.
\begin{itemize}
\item Approximating the optimal welfare within a fixed additive error~$\Delta$
is as hard as computing it exactly: given an instance, scale it by an
integer $K>2\Delta$, approximate the optimal welfare of the scaled instance
within~$\Delta$, divide by~$K$, and round to the nearest integer. Since the
original optimal welfare is integral and the error after division is
at most $\Delta/K<1/2$, this yields the exact value.
\item If a single instance of the class has a positive integrality gap,
the scaled instances have arbitrarily large integrality gaps and, by
Corollary~\ref{cor:nec}, arbitrarily large minimum regret.
\end{itemize}
Hence, on such classes, additive lower bounds of any size already follow
from a single instance with a positive gap and do not say anything about
the regret relative to the size of the valuations.
Part~\ref{item:relative-hardness} provides this scale-invariant conclusion.

%% file: acknowledgements.tex
\section*{Acknowledgements}

We thank Martin Bichler for helpful discussions and for providing
various pointers to the literature.
This research has been funded by the Deutsche Forschungsgemeinschaft
(DFG) in the project 543678993 (Aggregative gemischt-ganzzahlige
Gleichgewichtsprobleme: Existenz, Approximation und Algorithmen).
We acknowledge the support of the DFG.
Finally, the first and third authors thank the DFG for their support
within RTG 2126 ``Algorithmic Optimization''.


%% file: use-of-ai.tex
\section*{Use of AI}

During the revision of this manuscript, the authors used
  Anthropic's Claude and OpenAI's Codex to help implement revisions in
  the \LaTeX{} source, in particular to draft or rephrase passages, to
  write out intermediate steps in proofs and examples, to check
  calculations, and to compile the list of symbols in
  Appendix~\ref{app:symbols}. The authors reviewed and verified all of
  these changes and take full responsibility for the content of this
  article.


%% file: appendix.tex

\renewcommand{\theHfigure}{\theHsection.\arabic{figure}}
\section{Proof of Proposition~\ref{prop:directed-flow-gap}}
\label{app:directed-flow-gap}

\begin{proof}[Proof of Proposition~\ref{prop:directed-flow-gap}]
Draw the $n$ paths $P_i$ from the source $s_i=(0,n-i+1)$ horizontally to
$(i,n-i+1)$ and then vertically down to the destination $t_i=(i,0)$, for
$i=1,\ldots,n$.
Every pair of paths crosses exactly once. Subdivide at all crossings
and orient all horizontal segments to the right and all vertical
segments downwards. At each crossing replace the crossing vertex by
two vertices $v^-,v^+$ and one arc $v^-\to v^+$. Attach the incoming
west and north segments to $v^-$ and the outgoing east and south
segments to $v^+$. All arcs have capacity~1. This replacement preserves
planarity and acyclicity, and each $P_i$ remains a directed path.
Figure~\ref{fig:directed-flow} illustrates the construction for $n=4$.

The sources and destinations of any two players alternate on the
boundary of the rectangle containing the drawing. Consequently, two
paths connecting the respective source--destination pairs must meet, by
planarity (or the Jordan curve theorem). At an internal crossing gadget,
every path through either vertex uses the shared arc $v^-\to v^+$. At any
other internal vertex, the unique incoming and outgoing arcs are shared
whenever two paths meet. Sources and destinations cannot be internal
vertices of another path. Hence any two source--destination paths of
different players share an arc, so unit capacities forbid routing
both integrally. Each source has one outgoing unit-capacity arc, and
acyclicity rules out nonzero circulations. Every nonzero integral
strategy is therefore one source--destination path. Since routing any
$P_i$ is feasible, the integral welfare optimum is exactly~1.

For each player, put configuration weight $1/2$ on $P_i$ and weight
$1/2$ on the empty flow. Every shared arc belongs to exactly two of
the designated paths, and each remaining arc to one. Thus these
weights satisfy every capacity constraint and yield value $n/2$.\qedhere
\end{proof}

\begin{figure}[t]
\centering
\input{tikz-imgs/directed-flow-gap.tikz}
\caption{Directed split-grid construction for four players. Thick arcs
are the unit-capacity bottlenecks replacing crossings; all other arcs
also have unit capacity. Each designated path first goes right and then down.}
\label{fig:directed-flow}
\end{figure}
 
\section{List of Symbols}\label{app:symbols}
The following tables list the main symbols, grouped by the section in which
they are introduced. The last column refers to the definition of each symbol
and to the page on which it is defined.

\begingroup
\small
\def\RR{\raggedright\arraybackslash\hspace{0pt}}
\renewcommand{\arraystretch}{1.15}
\medskip
\noindent
\begin{tabular}{@{}p{0.27\textwidth}p{0.47\textwidth}p{0.18\textwidth}@{}}
\hline
Symbol & Meaning & Defined in\\
\hline
\multicolumn{3}{@{}l}{\emph{Section~\ref{sec:model}: Model}}\\[2pt]
\RR $N=\{1,\dots,n\}$ & \RR set of players & \RR Sec.~\ref{sec:model}, p.~\pageref{sym:players}\\
\RR $R=\{1,\dots,m\}$ & \RR set of resources (items) & \RR Sec.~\ref{sec:model}, p.~\pageref{sym:players}\\
\RR $u\in\Z_+^m$ & \RR capacities (multiplicities) of the resources & \RR Sec.~\ref{sec:model}, p.~\pageref{sym:players}\\
\RR $X_i$, $X=\times_{i\in N}X_i$ & \RR strategy space of player~$i$; strategy profiles & \RR Sec.~\ref{sec:model}, p.~\pageref{sym:strategies}\\
\RR $\ell(x)$, $\ell_j(x)=\sum_{i\in N}x_{ij}$ & \RR load vector of~$x$; load on resource~$j$ & \RR Sec.~\ref{sec:model}, p.~\pageref{sym:load}\\
\RR $X(u)$ & \RR capacity-feasible strategy profiles, $\ell(x)\leq u$ & \RR Sec.~\ref{sec:model}, p.~\pageref{sym:Xu}\\
\RR $\pi_i$, $\pi(x)=\sum_{i\in N}\pi_i(x_i)$ & \RR valuation of player~$i$; social welfare & \RR Sec.~\ref{sec:model}, p.~\pageref{sym:welfare}\\
\RR $\Lambda(x,u)$ & \RR complementary prices, $\lambda\geq0$ with $\lambda^\top(\ell(x)-u)=0$ & \RR Def.~\ref{def:walras}, p.~\pageref{def:walras}\\
\RR $\Reg_i(x_i,\lambda)$, $\Reg(x,\lambda)$ & \RR regret of player~$i$; total regret & \RR Sec.~\ref{sec:model}, p.~\pageref{sym:regret}\\
\RR $\Delta$ & \RR regret bound of a $\Delta$-regret \Walras equilibrium & \RR Def.~\ref{def:apx-implementable}, p.~\pageref{def:apx-implementable}\\
\RR \eqref{regret-opt} & \RR regret minimization problem & \RR p.~\pageref{regret-opt}\\
\RR \eqref{price-opt} & \RR social welfare problem & \RR p.~\pageref{price-opt}\\
\hline
\end{tabular}

\bigskip
\noindent
\begin{tabular}{@{}p{0.27\textwidth}p{0.47\textwidth}p{0.18\textwidth}@{}}
\hline
\multicolumn{3}{@{}l}{\emph{Section~\ref{sec:GenVal}: Characterizing Regret Equilibria}}\\[2pt]
\RR $\Req(x)$, $\Rleq(x)$ & \RR resources that are binding / have slack at~$x$ & \RR Sec.~\ref{sec:GenVal}, p.~\pageref{sym:Req}\\
\RR \eqref{price-opt-restricted}, \eqref{price-dual-min-restricted} & \RR $x$-parameterized welfare problem and its Lagrangian dual & \RR p.~\pageref{price-opt-restricted}\\
\RR $L_x$, $\mu_x$ & \RR $x$-parameterized Lagrangian and Lagrangian dual function & \RR Sec.~\ref{sec:GenVal}, p.~\pageref{sym:mux}\\
\RR $\gamma_x$, $\gamma_x^*$ & \RR $x$-parameterized duality gap function; duality gap & \RR \eqref{eq:gapfunction-restricted}, p.~\pageref{eq:gapfunction-restricted}\\
\RR \eqref{LP:price-opt-tilde-restricted}, $\alpha_i$, $\Simp_i$ & \RR $x$-parameterized configuration LP; configuration of player~$i$ and its simplex & \RR p.~\pageref{LP:price-opt-tilde-restricted}\\
\RR $\pi^{\LP}_x$, $\iota_x$, $\iota_x^*$ & \RR optimal value of~\eqref{LP:price-opt-tilde-restricted}; integrality gap function; integrality gap & \RR Def.~\ref{def:IntGap-restricted}, p.~\pageref{def:IntGap-restricted}\\
\RR $\mu$, $\gamma$, $\iota$, $\iota^*$, \eqref{price-dual-min} & \RR the corresponding objects for the classical problem~\eqref{price-opt} & \RR p.~\pageref{price-dual-min}\\
\RR $\bar X(u)$ & \RR strategy profiles with $\ell(x)=u$ & \RR \eqref{eq:strat-up}, p.~\pageref{eq:strat-up}\\
\RR $V_i$ & \RR maximum value $\max_{z_i\in X_i}\pi_i(z_i)$ of player~$i$ & \RR \eqref{eq:fractional-player-bound}, p.~\pageref{eq:fractional-player-bound}\\
\hline
\end{tabular}

\bigskip
\noindent
\begin{tabular}{@{}p{0.27\textwidth}p{0.47\textwidth}p{0.18\textwidth}@{}}
\hline
\multicolumn{3}{@{}l}{\emph{Section~\ref{sec:regret-bounds}: Regret Bounds}}\\[2pt]
\RR \eqref{LP:price-opt}, $\pi^{\LP}$ & \RR classical configuration LP and its optimal value & \RR p.~\pageref{LP:price-opt}\\
\RR $\rho(q)$ & \RR optimal value of~\eqref{LP:price-opt} with capacity vector~$q$; $\rho(u)=\pi^{\LP}$ & \RR Sec.~\ref{sec:regret-bounds}, p.~\pageref{sym:rho}\\
\RR $s(x)=u-\ell(x)$, $r(x)$ & \RR slack vector of~$x$; minimum regret $\min_{\lambda\in\Lambda(x,u)}\Reg(x,\lambda)$ at~$x$ & \RR Sec.~\ref{sec:weighted-seller}, p.~\pageref{sym:slack}\\
\RR $S(x,\lambda)=\lambda^\top s(x)$ & \RR seller regret & \RR Sec.~\ref{sec:weighted-seller}, p.~\pageref{sym:seller}\\
\RR $F_\tau(x)$ & \RR minimum of buyer regret plus $\tau$ times seller regret & \RR \eqref{eq:weighted-regret}, p.~\pageref{eq:weighted-regret}\\
\RR $q_\tau(x)=\ell(x)+\tau s(x)$ & \RR adjusted capacity vector & \RR Thm.~\ref{thm:weighted-regret}, p.~\pageref{thm:weighted-regret}\\
\RR $v_i(\lambda)$ & \RR best-response utility of player~$i$ at prices~$\lambda$ & \RR Sec.~\ref{sec:weighted-seller}, p.~\pageref{sym:vi}\\
\RR $M_j$ & \RR maximum aggregate consumption of resource~$j$ & \RR \eqref{eq:max-aggregate}, p.~\pageref{eq:max-aggregate}\\
\RR $\tau_0(x)$ & \RR penalty threshold & \RR \eqref{eq:penalty-threshold}, p.~\pageref{eq:penalty-threshold}\\
\RR $u_{\max}=\max_{j\in R}u_j$ & \RR largest capacity & \RR \eqref{eq:uniform-penalty-threshold}, p.~\pageref{eq:uniform-penalty-threshold}\\
\RR $\eta(x)$, $(a)_+=\max\{a,0\}$ & \RR slack factor; positive part & \RR \eqref{eq:slack-factor}, p.~\pageref{eq:slack-factor}\\
\hline
\end{tabular}

\bigskip
\noindent
\begin{tabular}{@{}p{0.27\textwidth}p{0.47\textwidth}p{0.18\textwidth}@{}}
\hline
\multicolumn{3}{@{}l}{\emph{Section~\ref{sec:algorithms}: Polynomial-Time Approximation Algorithms}}\\[2pt]
\RR $L=\langle I\rangle$ & \RR encoding length of instance~$I$ & \RR Sec.~\ref{sec:algorithms}, p.~\pageref{sym:encoding}\\
\RR \eqref{LP-dual} & \RR dual of the classical configuration LP & \RR p.~\pageref{LP-dual}\\
\RR $\ALG(I)$, $\alpha(I)$ & \RR allocation returned by $\ALG$; additive configuration-LP gap guarantee & \RR Def.~\ref{def:additive-verification}, p.~\pageref{def:additive-verification}\\
\hline
\end{tabular}

%% file: tikz-imgs/directed-flow-gap.tikz
\begin{tikzpicture}[scale=1.05,>=stealth,draw=red,text=red]
\draw[thin,->] (1,4) -- (0.88,3.12);
\draw[thin,->] (2,3) -- (1.88,2.12);
\draw[thin,->] (3,2) -- (2.88,1.12);
\draw[thin,->] (4,1) -- (4,0);
\draw[very thick,->] (0.88,1.12) -- (1.12,0.88);
\draw[very thick,->] (0.88,2.12) -- (1.12,1.88);
\draw[very thick,->] (0.88,3.12) -- (1.12,2.88);
\draw[very thick,->] (1.88,1.12) -- (2.12,0.88);
\draw[very thick,->] (1.88,2.12) -- (2.12,1.88);
\draw[very thick,->] (2.88,1.12) -- (3.12,0.88);
\draw[thin,->] (1.12,0.88) -- (1.88,1.12);
\draw[thin,->] (1.12,0.88) -- (1,0);
\draw[thin,->] (1.12,1.88) -- (0.88,1.12);
\draw[thin,->] (1.12,1.88) -- (1.88,2.12);
\draw[thin,->] (1.12,2.88) -- (2,3);
\draw[thin,->] (1.12,2.88) -- (0.88,2.12);
\draw[thin,->] (2.12,0.88) -- (2.88,1.12);
\draw[thin,->] (2.12,0.88) -- (2,0);
\draw[thin,->] (2.12,1.88) -- (3,2);
\draw[thin,->] (2.12,1.88) -- (1.88,1.12);
\draw[thin,->] (3.12,0.88) -- (4,1);
\draw[thin,->] (3.12,0.88) -- (3,0);
\draw[thin,->] (0,4) -- (1,4);
\draw[thin,->] (0,3) -- (0.88,3.12);
\draw[thin,->] (0,2) -- (0.88,2.12);
\draw[thin,->] (0,1) -- (0.88,1.12);
\node[left] at (0,4) {$s_1$};
\node[below] at (1,0) {$t_1$};
\node[left] at (0,3) {$s_2$};
\node[below] at (2,0) {$t_2$};
\node[left] at (0,2) {$s_3$};
\node[below] at (3,0) {$t_3$};
\node[left] at (0,1) {$s_4$};
\node[below] at (4,0) {$t_4$};
\end{tikzpicture}